%% file: main.tex
\documentclass[3p,times,review,procedia,authoryear]{elsarticle}
\usepackage{amssymb}
\usepackage{amsmath}
\usepackage{mathrsfs}
\usepackage{graphicx} 
\usepackage{subfigure} 
\usepackage{float} 
\usepackage{array} 
\usepackage{booktabs} 
\usepackage{setspace} 
\usepackage[colorlinks=true, allcolors=blue]{hyperref} 
\usepackage{enumerate} 
\usepackage{threeparttable} 
\numberwithin{equation}{section} 

\newtheorem{theorem} {Theorem}[section]
\newtheorem{lemma} {Lemma} [section]
\newtheorem{proof} {Proof}[section]
\newtheorem{proposition} {Proposition}[section]

\numberwithin{figure}{section}
\numberwithin{table}{section}

\usepackage[figuresright]{rotating}

\begin{document}

\begin{frontmatter}



\title{ Copula-based semiparametric nonnormal transformed linear model for survival data with dependent censoring }




\author[author1,author2]{Huazhen Yu\corref{cor1}}
\cortext[cor1]{Corresponding author}
\ead{11835014@zju.edu.cn}

\author[author1,author3]{Lixin Zhang}
\ead{stazlx@zju.edu.cn}

\address[author1]{School of Mathematical Sciences, Zhejiang University, Hangzhou, Zhejiang 310027, China}
\address[author3]{School of Statistics and Mathematics, Zhejiang Gongshang University, Hangzhou 310018,  Zhejiang, China}


\begin{abstract}


Although the independent censoring assumption is commonly used in survival analysis, 
it  can be violated when the censoring time is related to the survival time, which often happens in many practical applications.
To address this issue,  we propose a flexible semiparametric method for dependent censored data. 
Our approach involves fitting the survival time and the censoring time with a joint transformed linear model, where the transformed function is unspecified. This allows for a very general class of models that can account for possible covariate effects, while  also accommodating administrative censoring.
We assume that the transformed variables have a bivariate nonnormal distribution based on parametric copulas and parametric marginals, which further enhances the flexibility of our method.
We demonstrate the identifiability of the proposed model and establish the consistency and asymptotic normality of the model parameters under appropriate regularity conditions and assumptions. 
Furthermore, we evaluate the performance of our method through extensive simulation studies, and provide a real data example for illustration.

\end{abstract}


\end{frontmatter}



\vspace{-8mm}
\input{sections/Sec1_introduction}

\input{sections/Sec2_model}


\input{sections/Sec3_esitimation}


\input{sections/Sec4_properties}


\input{sections/Sec5_simulation}


\input{sections/Sec6_real}


\input{sections/Sec7_conclusion}

\appendix

\renewcommand{\thelemma}{D.\arabic{lemma}} 

\input{sections/AppendixA_identification}

\input{sections/AppendixB_consistency}

\input{sections/AppendixC_normality}

\input{sections/AppendixD_Lemma}

\section*{Acknowledgment}
\vspace{-2mm}
 This work wad supported by grants from  NSF of China (Grant Nos.U23A2064 and 12031005).
 \vspace{-5mm}

\bibliographystyle{elsarticle-harv} 

\bibliography{sections/refer}















\end{document}

%% file: sections/Sec1_introduction.tex
\section{Introduction}\label{sec:intro}

Survival analysis is a widely used statistical technique in various fields, such as medicine, engineering, and social sciences. 
In survival analysis, the existence of right censoring happens quite often, for example, at the end of the study period or the withdrawal from the study.
In such situations, investigators only  observe the minimum of survival time $T$ and censor time $C$.
For simplicity, it is common to assume that  $T$ and $C$ are independent of each other and then conduct model estimation. 
However, in many applications, such independence may not be reasonable, such as \citet{LAGAKOS1978} and \citet{Seungyeoun1995}. 
They  considered the data set included 61 inoperable lung cancer patients treated with cyclophosphamide. Some patients dropped out of the study due to worsening condition or adverse effects requiring different treatments. Thus, in this case, the censoring time was related to the survival time and we call it dependent censor.
To distinguish from the dependent censor, the independent censor due to the end of the observation period is called as administrative censor.

Several methods have been proposed in the literature for analysis of the relationship between $T$ and $C$ with dependent censoring data.
For instance,  methods based on  known copula have been explored by \citet{ZHENG1995} and  \citet{RIVEST2001138}.
These approaches have been further applied to different regression models, as studied by  \citet{Braekers2005}, \citet{Huang2008}, \citet{Chen2010} and \citet{SUJICA201889}. 
Another important approach involves  the parametric models, as demonstrated by \citet{BASU1978413}, \citet{Emoto1990}, \citet{Deresa2019}.
Recently, \citet{Deresa2021} introduced a semiparametric transformation regression model that allows for the analysis of covariates. 
However, this method is limited by the assumption that the error vector follows a standard bivariate normal distribution. 
\citet{Deresa2022}  investigated  a more flexible copula approach for bivariate survival data subject to left truncation and dependent censoring, building upon a copula-based method proposed by \citep{Czado2023}.
However, this method is entirely parametric, relying on parametric copulas and  marginal distributions, which requires parametric specification of both the joint distribution and the relationship between $T, C$ and covariates.
Misspecification in such cases often leads to biased estimators.
Therefore,  it is crucial to to develop a method that relaxes the restriction on parametric specification without assumption of bivariate normality.

For this task, we propose a joint transformed linear regression model with  unknown transformations for survival data subject to dependent censoring.
This continued the research direction established by \citet{Deresa2021} and \citet{Deresa2022} and further developed it.
We model the relationship between $T$ and $C$ using a parametric copula with semiparametric  marginal distribution.
At first sight, this model appears to be an ordinary joint regression problem. 
However, the challenge lies in dealing with the identifiability of the model in the presence of unbounded and unspecified transformation function and nonnormal error vector. 
The approach for the nonnormal transformated linear model is motivated by \citet{Tong2019}, who developed a method to estimate the unknown $H$  for uncensored data with a general nonnormal error distribution.
\citet{Yu2024} applied this method to dependent current status data, establishing the identifiability of their proposed model.
In this paper, we extend these ideas to handle data subject to both dependent and administrative censoring. The identification challenge is particularly significant, as we only observe the minimum of $T$, $C$, and $A$, but never all three simultaneously. To address this, we adopt concepts from \citet{Tong2019} and \citet{Deresa2021} to construct a nonparametric estimator for the transformation function. Moreover, under appropriate regularity conditions and assumptions, we establish the consistency and asymptotic normality of the parameter estimators.

The contributions of our approach can be summarized in three main aspects.
Firstly, 
our method focuses on situations where the vector of errors follows a nonnormal bivariate distribution.
The dependence structure and the marginal distributions are described by a parametric class of copulas and a parametric  class of nonnormal distributions, respectively.
A crucial issue in this context is the identifiability of the model, especially the association parameter of copula function.
To address this, we present the sufficient assumptions on parametric copulas and marginal distributions to achieve identification.
Second, 
despite the unboundness of transformation, we present a nonparametric estimator of transformation function without any specification. We demonstrate the consistency of the resulting estimators.
Third, 
our approach employs the maximum likelihood method, enabling simultaneous estimation of  the model parameters and transformation function.
Furthermore, we establish consistency and asymptotic normality of the parameter estimators under some conditions.

The paper is organized as follows. In Section \ref{sec:model}, we provide a general description of the model and the distribution of the observations  on which our approach is built.
We finish Section  \ref{sec:model} with our identifiability result and the sufficient  assumptions.
Section  \ref{sec:estimation}  develops the maximum likelihood estimation approach for our model, 
while in Section \ref{sec:property} we establish the asymptotic properties.
Section \ref{sec:simulation} presents several simulation studies and a formal goodness-of-fit test to assess the finite sample performance of our method, demonstrating its effectiveness in practical situations.
In Section  \ref{sec:real}, we apply the proposed method to a real data example.
Section \ref{sec:conclusion} provides a discussion and concludes the paper, including future research directions.
Finally, the proofs of the theorems are provided in the appendix.

%% file: sections/Sec2_model.tex
\section{Model   Specification}\label{sec:model}

\subsection{Linear transformation model}

Similar to \cite{Huang2008} and \cite{Deresa2021}, we also let $T$ denote the survival time of interest, $C$ denote the dependent censoring time, and $A$ denote the administrative censoring time.
The variables $T$, $C$, and $A$ all take values within the range $-\infty$ to $\infty$.
We represent $X$ and $W$ as the vectors of covariates associated with $T$ and $C$, respectively.
These covariates can be identical, partially overlapping, or entirely different. 
We assume that, given $X$ and $W$, both $T$ and $C$ are conditionally independent of $A$.
However, $T$ and $C$ may still exhibit dependence even after accounting for the effect of covariates $X$ and $W$. 
Consequently, we consider a linear transformation model for $T$ and $C$ as follows:
\begin{equation}
  \begin{cases}
  H(T)=X^{T}\beta+\varepsilon_{T},\quad  \\
  H(C)=W^{T}\eta+\varepsilon_{C},\quad
  \end{cases} 
  \label{transmodel}
\end{equation}
where   $\beta$ is a $p$-dimensional vector of regression coefficients, and $\eta$ is a $q$-dimensional vector of regression coefficients.
Here, $H \in \mathcal{H} $ is an unknown strictly increasing function. 
Without loss of
generality, we assume $H(0) = 0$.
Since $A$ denotes independent censoring time, the distribution of $A$ is left unspecified here. 
Then we assume the vector of error terms $(\varepsilon _{T} ,\varepsilon _{C} ) $ follows a copula-based bivariate  distribution:
  \begin{equation}
	P(\varepsilon _T<t,\varepsilon _C<c)
	=\mathcal{C} (F_{\varepsilon _T }(t),F_{\varepsilon _C }(c) ), 
 \label{eq22}
   \end{equation}
	where $F_{\varepsilon _T}$ and $F_{\varepsilon _C}$ are the distributions  of $\varepsilon _T$ and $\varepsilon _C$.
 Both the copula function  and the marginal distributions are assumed to be parametric, defined as follows:
 \begin{equation}
	 \mathcal{C}\in \{ \mathcal{C}_r, r\in \Upsilon  \},\quad
	 F_{\varepsilon _T} \in \{ F_{\varepsilon _T,\lambda_T}, \lambda_T \in \Lambda_T \} \quad and \quad
	 F_{\varepsilon _C} \in \{ F_{\varepsilon _C,\lambda_C}, \lambda_C \in \Lambda_C \},
  \label{eq23}
\end{equation}
where $\Upsilon$, $\Lambda_T$ and $\Lambda_C$ are the parameter spaces.	
	Let 
	$$\mathcal{C} '_{1} (u,v)=\frac{\partial }{\partial u} \mathcal{C}(u,v), \quad
	\mathcal{C} '_{2} (u,v)=\frac{\partial }{\partial v} \mathcal{C}(u,v).
	$$
	Define the conditional distribution functions by
	$$
	F_{\varepsilon _T|\varepsilon _C}(t|c)
	=P(\varepsilon _T<t|\varepsilon _C=c)
        =\frac{P(\varepsilon _T<t,\varepsilon _C=c)}{P(\varepsilon _C=c)}
	=\frac{\mathcal{C} '_{2} (F_{\varepsilon _T }(t),F_{\varepsilon _C }(c)) f_{\varepsilon _C}(c) }{f_{\varepsilon _C}(c)}
         =\mathcal{C} '_{2} (F_{\varepsilon _T }(t),F_{\varepsilon _C }(c)),
	$$
	
	$$
	F_{\varepsilon _C|\varepsilon _T}(c|t)
	=P(\varepsilon _C<c|\varepsilon _T=t)
        =\frac{P(\varepsilon _C<c,\varepsilon_T=t )}{P(\varepsilon _T=t)}
	=\frac{\mathcal{C} '_{1} (F_{\varepsilon _T }(t),F_{\varepsilon _C }(c)) f_{\varepsilon _T}(t) }{f_{\varepsilon _T}(t)}
        =\mathcal{C} '_{1} (F_{\varepsilon _T }(t),F_{\varepsilon _C }(c)),
	$$
where $f_{\varepsilon _T}$ and $f_{\varepsilon _C}$ are the density functions  of $\varepsilon _T$ and $\varepsilon _C$, respectively.

\subsection{Distribution of the observation}

Due to right censoring, 
the observed time is defined as  $Z = \min(T, C, A)$. The censoring indicators are defined as $\Delta = I(Z = T)$ and $\xi = I(Z = C)$, where $I(\cdot)$ denotes the indicator function. Thus, the observed data take the form $\{ (Z_i, \Delta_i, \xi_i, X_i, W_i) , i = 1, \ldots, n \}$.
In the case where $Z\equiv T$ (i.e. $\Delta=1, \xi=0$ ),  we have
	$$
	\begin{aligned}  
	F_{Z,\Delta,\xi |X,W}(z,1,0 |x,w)
	&=P(Z\le z,\Delta=1,\xi=0|X=x,W=w) \\  
	&=P(Z=T\le z,T<C,T<A|X=x,W=w) \\
	&=P(T\le z,H(T)<H(C),T<A|X=x,W=w)\\
	&=\int_{-\infty}^{z} 
	P(\varepsilon _C>H(t)-w^T\eta|\varepsilon _T=H(t)-x^T\beta))P(A>t)f_{T}(t)dt , \\  
	\end{aligned}
	$$
	$$
	\begin{aligned} 
	f_{Z,\Delta,\xi |X,W}(z,1,0 |x,w)
	=f_{\varepsilon _T}(H(z)-x^T\beta)
	    S_{\varepsilon _C|\varepsilon _T} (H(z)-w^T\eta | H(z)-x^T\beta)
	    P(A>z) h(z), 
	\end{aligned}
	$$
 where $S_{\varepsilon _C|\varepsilon _T} (c | t)=1-F_{\varepsilon _C|\varepsilon _T} (c | t)$, $h(z)=dH(z)/dz$.
	Similarly, when  $Z\equiv C$ (i.e. $\Delta=0, \xi=1$ ), we have:
	$$
	\begin{aligned}  
	f_{Z,\Delta,\xi |X,W}(z,0,1 |x,w)
	=f_{\varepsilon _C}(H(z)-w^T\eta)
	    S_{\varepsilon _T|\varepsilon _C}  (H(z)-x^T\beta | H(z)-w^T\eta)
	    P(A>z) h(z).
	\end{aligned}
	$$
	Finally, when  $Z\equiv A$ (i.e. $\Delta=0, \xi=0$ ), we have:
	$$
	\begin{aligned}  
	f_{Z,\Delta,\xi |X,W}(z,0,0 |x,w)
	&=P(T>z,C>z|X=x,W=w)f_A(z)\\
	&=\tilde{\mathcal{C}}( F_{\varepsilon _T} (H(z)-x^T\beta) ,  F_{\varepsilon _C}(H(z)-w^T\eta))f_A(z),
	\end{aligned}
	$$
where $\tilde{ \mathcal{C}}(u,v)=1-u-v + \mathcal{C}(u,v)$
and $f_A$ is the density of $A$.
In order to demonstrate the identifiability of model (\ref{transmodel})-(\ref{eq23}) at the end of this section, we  make the following basic assumptions.

\begin{enumerate}[($\mathnormal{A}$1)]\setlength{\itemsep}{-4pt}
    \item\label{A1}
    The error terms ($\varepsilon_T$, $\varepsilon_C$) are independent of the covariates $(X, W)$.
    \item\label{A2}   
    Given $(X, W)$, the times $(T, C)$ are independent of $A$.
    \item\label{A3}   The probabilities $P(Z = T)$, $P(Z = C)$,  $P(Z = A)$  are all strictly positive.
    \item\label{A4} 
    Both \( \text{var}(X) \) and \( \text{var}(W) \) are full-rank matrices.
    \item\label{A5}  
    The distribution of $A$ is independent of all model parameters.
    \item\label{A8}   
     $\mathcal{H}$ consists of all non-decreasing and differentiable functions on $(-\infty,\infty)$, satisfying $0 < \lim_{z\to -\infty} \{h(z)k(z)\} < \infty$ for some fixed function $k: (-\infty,\infty) \to [0,\infty)$ related to the true $H$ function, where $h(t) = dH(t)/dt$. 
\end{enumerate}

The above assumptions concerning the variables $T, C, A$ and covariates $X, W$ are essential and required for model identifiability. 
Notably,  \cite{Deresa2021} was the first to introduce these assumptions, including (A\ref{A8}) on the transformed function $H$, to ensure the identifiability of the transformed linear model.
For Assumption (A\ref{A8}), the function $k(z)$ controls the behavior of the transformation function $H(z)$. 
Specifically, Assumption (A\ref{A8})  is used in the initial step of proving Theorem \ref{theorem:identification}.
The form of $k(z)$ depends on the true transformation function, such as a power function $k(z) = z^c$ or a logarithmic function $k(z) = \log(1+z^c)$, where $c$ is a constant. 
In order to achieve identifiability under the nonnormal error term, we need some additional assumptions on the dependence structure and marginal distribution.

\begin{enumerate}[($\mathnormal{A}$1)]\setlength{\itemsep}{-4pt}
\setcounter{enumi}{6}

    \item \label{A5+} The functions $g_{\varepsilon _T,\lambda_T}(t)=  \partial  \log f_{\varepsilon _T,\lambda_T}(t)/\partial t $, 
    $g_{\varepsilon _C,\lambda_C}(t)=  \partial \log f_{\varepsilon _C,\lambda_C}(t)/\partial t $ are strictly decreasing with $g_{\varepsilon _T,\lambda_T}(0) = 0$, $g_{\varepsilon _C,\lambda_C}(0) = 0$ and the limits exist as $t \to a$ for  $a=-\infty$ and $\infty$.

    \item\label{A6}  
    For all  $x$, $w$ and for all $(\lambda_T,\lambda_C,r)\in \Lambda_T \times \Lambda_C \times \Upsilon$, 
    $a = -\infty$ or $+\infty$, 
    we have
    $$
	\lim_{t\to a} F_{\varepsilon _T|\varepsilon _C,\lambda_{T},\lambda_{C},r}(H(t)-x^T\beta|H(t)-w^T\eta)=0 , 
    $$
$$
	\lim_{t \to a} F_{\varepsilon _C|\varepsilon _T,\lambda_{T},\lambda_{C},r}(H(t)-w^T\eta|H(t)-x^T\beta)=0 .
$$

    \item\label{A7}   
    Consider functions $g_{\varepsilon _T,\lambda_T}(t)$ and $g_{\varepsilon _C,\lambda_C}(t)$ defined in (A\ref{A5+}).
    For all $\lambda_{T,1},\lambda_{T,2} \in \Lambda_T$ and  for $a = - \infty$ and $+\infty$, we have
$$ \lim_{t\to a}  \frac{g_{\varepsilon _T,\lambda_{T,1}}(t)}{g_{\varepsilon _T,\lambda_{T,2}}(t)}=1
\quad
\Longleftrightarrow  \lambda_{T,1}=\lambda_{T,2}
\quad or \quad 
\lim_{t\to a}  \frac{f_{\varepsilon _T,\lambda_{T,1}}(t)}{f_{\varepsilon _T,\lambda_{T,2}}(t)}=1 .$$

 Similarly, for all $\lambda_{C,1},\lambda_{C,2} \in \Lambda_C$ and  for $a = -\infty$ and $+\infty$, we have
$$\lim_{t\to a}  \frac{g_{\varepsilon _C,\lambda_{C,1}}(t)}{g_{\varepsilon _C,\lambda_{C,2}}(t)}=1
\quad 
\Longleftrightarrow  \lambda_{C,1}=\lambda_{C,2}
\quad or \quad 
\lim_{t\to a}  \frac{f_{\varepsilon _C,\lambda_{C,1}}(t)}{f_{\varepsilon _C,\lambda_{C,2}}(t)}=1.$$

\item\label{A9}
Consider functions $H_1$ and $H_2 \in \mathcal{H}$. Let $0 < C_a < \infty$ be a constant that depends on $a$, where $a = -\infty$. For all $\lambda_{T,1}, \lambda_{T,2} \in \Lambda_T$, if
$$
\lim_{t \to a} \frac{H_1(t)}{H_2(t)}=\frac{1}{C_a}, \quad \text{and} \quad
\lim_{t \to a} \frac{f_{\varepsilon_T, \lambda_{T,1}}(H_1(t) - x^T \beta_1)}{f_{\varepsilon_T, \lambda_{T,2}}(H_2(t) - x^T \beta_2)} = C_a, \, \text{for all } x ,
$$
then both limits hold if and only if
 $ C_a=1$.

Similarly, for all $\lambda_{C,1}, \lambda_{C,2} \in \Lambda_C$,  if
$$
\lim_{t \to a} \frac{H_1(t)}{H_2(t)}=\frac{1}{C_a}, \quad \text{and} \quad
\lim_{t \to a} \frac{f_{\varepsilon_C, \lambda_{C,1}}(H_1(t) - w^T \eta_1)}{f_{\varepsilon_C, \lambda_{C,2}}(H_2(t) - w^T \eta_2)} = C_a, \, \text{for all } w
$$
then both limits hold if and only if
 $ C_a=1$.
When $a=\infty$, the above assumption still hold.

\end{enumerate}

Assumptions (A\ref{A5+}) and (A\ref{A7}) pertain to the marginal densities of $T$ and $C$. They are similar to but stricter than Assumption (A1) in \citet{Deresa2022}.
Since our model involves the nonparametric estimator of $H$, these higher-level assumptions are required.
It is worth noting that a popular class of density functions, given by
\begin{equation}
    f_{\lambda}(x)=\frac{e^x}{(1+\lambda e^x)^{1+1/\lambda}},
\label{eq_tong}
\end{equation}
satisfies Assumptions (A\ref{A5+}) and (A\ref{A7}), as proven in \cite{Tong2019}. 
This class of density gives the Cox model and the proportional odds model with $\lambda = 0$ and $1$, respectively.
Furthermore, Theorem 3.2 of \citet{Deresa2022} demonstrates that Assumption (A\ref{A6}) is satisfied for many parametric families of copulas and densities under certain conditions.
According to  \citep{Czado2023},
Assumption (A\ref{A6}) is a sufficient condition for identification but not a necessary one.
Assumption (A\ref{A9}) pertains to both
the margins  family and transformed functions. 
We can verify  that (A\ref{A9}) also holds for the parametric marginal functions (\ref{eq_tong}) through  the following Proposition \ref{prop:A10}.
Indeed, several commonly used density functions satisfy this assumption. For example, the Weibull distribution preserves this equivalence.
However, the normal density function $f_{\sigma}(x)=exp(-x^2/2\sigma^2)/\sqrt{2\pi \sigma^2} $ of the distribution \( N(0,\sigma^2) \) does not satisfy this assumption. It is important to exercise caution when validating (A\ref{A9}), as the introduction of a parameter can influence the result, even within the same family of distributions. The Cauchy distribution serves as an example. While the standard Cauchy distribution, with density function $f(t) = 1/(\pi(1+t^2))$, clearly satisfies (A\ref{A9}), the Cauchy distribution with a scale parameter $\lambda$, i.e., $f(t) = \lambda/\pi(\lambda^2 + t^2)$, can only lead to $C_a = \lambda_2/\lambda_1$.

\begin{proposition}
    Assumption (A\ref{A9}) holds for $\varepsilon $ that follows the density  $f_{\lambda}(x)$, which is given in
    (\ref{eq_tong}).
\label{prop:A10}
\end{proposition}

\begin{proof}
    By letting $\lim_{t\to a} \{H_1(t)/H_2(t)\}=1/C_a $, we can derive that
    $$
    \lim_{t\to a}  \frac{f_{\lambda_{1}}(H_1(t)-x^T \beta_1)}{f_{\lambda_{2}}(H_2(t)-x^T \beta_2)}
    =\lim_{t\to a}
    \frac{e^{H_1(t)-x^T \beta_1}}{e^{H_2(t)-x^T \beta_2}}
    \frac{(1+\lambda_2 e^{H_2(t)-x^T \beta_2})^{1+1/\lambda_2}}
    {(1+\lambda_1 e^{H_1(t)-x^T \beta_1})^{1+1/\lambda_1}}.
    $$
    We first consider $a=-\infty$ 
    and  have
    \begin{equation}
        \lim_{t\to a} 
    \frac{f_{\lambda_{1}}(H_1(t)-x^T \beta_1)}{f_{\lambda_{2}}(H_2(t)-x^T \beta_2)}
    =\lim_{t\to a} exp((H_1(t)-H_2(t))+x^T(\beta_2-\beta_1)), 
    \label{A9proof1}
    \end{equation}
    Since $\lim_{t\to a} \{H_1(t)/H_2(t)\}=1/C_a$, we can deduce that:
    $$\lim_{t\to a} \frac{H_1(t)-H_2(t)}{H_2(t)}=\frac{1}{C_a}-1, $$
    which implies that when $ C_a \ne 1$, 
    $\lim_{t\to a} \{H_1(t)-H_2(t)\}$ is $-\infty$ or $\infty$.
    Since $0<C_a<\infty$, (\ref{A9proof1}) equals constant $C_a$ only when $\lim_{t\to -\infty} \{H_1(t)-H_2(t)\} \ne -\infty$ or $\infty$ and therefore $C_a=1$.
    
    Similarly, we consider $a=+\infty$ and we  have
    \begin{equation}
    \begin{aligned}
        \lim_{t\to a} 
    \frac{f_{\lambda_{1}}(H_1(t)-x^T \beta_1)}{f_{\lambda_{2}}(H_2(t)-x^T \beta_2)}
    =\lim_{t\to a}
    \frac{\lambda_2^{1+1/\lambda_2}}{\lambda_1^{1+1/\lambda_1}}
    exp(\frac{H_2(t)}{\lambda_2} - \frac{H_1(t)}{\lambda_1}) exp(x^T(\frac{\beta_1}{\lambda_1} - \frac{\beta_2}{\lambda_2}   ) ),     
     \end{aligned}
     \label{A9proof2}
    \end{equation}
    and
    $$\lim_{t\to a} \frac{H_2(t)/\lambda_2-H_1(t)/\lambda_1}{H_2(t)}=\frac{1}{ \lambda_2}-\frac{1}{C_a \lambda_1}. $$
    Thus (\ref{A9proof2}) equals $C_a$ only when there exist a constant $-\infty<D<\infty$ such that
    $\lim_{t \to a} \{ H_2(t)/\lambda_2 - H_1(t) /\lambda_1\}=D$, which means
    $\lambda_2=C_a \lambda_1$. 
    Then 
    if $D<0$, we have $\lim_{t \to a}H_2(t)< \lim_{t \to a} C_a H_1(t) $, 
    and 
    $$
    \frac{1}{C_a }=
    \lim_{t \to a} \frac{H_1(t)}{H_2(t)}>
    \lim_{t \to a}
    \frac{H_1(t)}{C_a H_1(t)}=
    \frac{1}{C_a },
    $$
    which is impossible. 
    Similarly, if $D>0$,  we can also obtain  a contradiction.
    Thus $D=0$.
    Then from (\ref{A9proof2}) we have  
    $$
    \lim_{t\to a}
    exp(x^T(\frac{\beta_1}{\lambda_1} - \frac{\beta_2}{\lambda_2}   ) )
    = \frac{\lambda_1^{1/\lambda_1}}{\lambda_2^{1/\lambda_2}}, \,
    \text{ for  almost every x,}
    $$
    which means that $Var(x^T(\frac{\beta_1}{\lambda_1} - \frac{\beta_2}{\lambda_2}))=0$.
     It now follows from Assumption (A\ref{A4}) that $\beta_1 / \lambda_1 - \beta_2/\lambda_2=0$ and hence we have 
     $ \lambda_1^{1/\lambda_1}=\lambda_2^{1/\lambda_2}$.
     Therefore, we have $\lambda_1=\lambda_2$ and $C_a=1$.
    
\end{proof}

Considering the parameters $(\theta, H) = (\beta, \eta, \lambda_T, \lambda_C, r, H)$, the model identifiability can be summarized as follows.

\begin{theorem}
\label{theorem:identification} (Identifiability of the Model):
Suppose the Assumptions (A\ref{A1})-(A\ref{A9}) given above hold true.
Consider the model (\ref{transmodel})-(\ref{eq23}) defined by the parameters $ (\beta,\eta,\lambda_T,\lambda_C,r, H)$, where $\beta \in \mathcal{R}^p, \eta \in \mathcal{R}^q, \lambda_T \in \Lambda_T , \lambda_C \in \Lambda_C , r \in \Upsilon , H\in  \mathcal{H}$.
Let $(\theta_i, H_i)$ be two sets of parameters, where $\theta_j=(\beta_j, \eta_j, \lambda_{T,j}, \lambda_{C,j}, r_j)$ for $j=1,2$. 
Let $f_{Z,\Delta,\xi |X,W}(Z_j, \Delta_j, \xi_j | X, W; \theta_j, H_j)$ denote the probability density function of the observed data given $(X, W)$ under the parameters $(\theta_i, H_i)$.
If
$$
f_{Z,\Delta,\xi |X,W}(Z_1, \Delta_1, \xi_1 | X, W; \theta_1, H_1) = f_{Z,\Delta,\xi |X,W}(Z_2, \Delta_2, \xi_2 | X, W; \theta_2, H_2) 
$$
for all $(X, W)$,
we have $(\theta_1, H_1) = (\theta_2, H_2)$,
which implies that the model is identifiable.
\end{theorem}
To prove the identifiability of the model, we follow the similar steps of the proof used in Theorem 1 of \cite{Deresa2021}. The detailed proof is provided in \ref{sec:proof31}.

%% file: sections/Sec3_esitimation.tex
\section{Model estimation}
\label{sec:estimation}

\subsection{Estimation of parameters}

Let $\theta=(\beta,\eta,\lambda_T,\lambda_C
	,r) \in \Theta$, where 
 $\Theta = \mathcal{R}^{p+q}\times \Lambda_T \times \Lambda_C \times \Upsilon $. When $Z_i\equiv T_i$ (i.e. $\Delta_i=1, \xi_i=0$ ),  we have the corresponding subdensity  is
	$$
	\begin{aligned}  
	L_{i1}(\theta,H)
	\propto 
	f_{\varepsilon _T}(H(Z_i)-X_i^T\beta)
	S_{\varepsilon _C|\varepsilon _T} (H(Z_i)-W_i^T\eta | H(Z_i)-X_i^T\beta)
	\mathrm{d} H(Z_i),
	\end{aligned}
	$$
 where $\propto $ represents proportional.
	Similarly, when  $Z_i \equiv C_i$ (i.e. $\Delta_i=0, \xi_i=1$ ),
	$$
	\begin{aligned}  
	L_{i2}(\theta,H)
	\propto
	f_{\varepsilon _C}(H(Z_i)-W_i^T\eta)
	 S_{\varepsilon _T|\varepsilon _C}  (H(Z_i)-X_i^T\beta | H(Z_i)-W_i^T\eta)
	 \mathrm{d} H(Z_i).
    \end{aligned}
	$$	
	Finally, when  $Z_i\equiv A_i$ (i.e. $\Delta_i=0, \xi_i=0$ ),
	$$
	\begin{aligned}  
	L_{i3}(\theta,H) 
	\propto
	\tilde{\mathcal{C}}(F_{\varepsilon _T} (H(Z_i)-X_i^T\beta ), F_{\varepsilon _C}( H(Z_i)-W_i^T\eta)).
	\end{aligned}
	$$
Therefore, the likelihood corresponding to the $i$th observation can be expressed as
	$$
	L_{i}(\theta,H)
	\propto
	 L_{i1}(\theta,H)^{\Delta_i(1-\xi_i)}
	 L_{i2}(\theta,H)^{(1-\Delta_i)\xi_i}
     L_{i3} (\theta,H)^{(1-\Delta_i)(1-\xi_i)}.
	$$

The estimator $\hat{s}_n = (\hat{\theta}_n, \hat{H}_n)$ is obtained by maximizing the likelihood function ${\textstyle \prod_{i=1}^{n}} L_i(\theta, H)$ over the parameter space $\Theta \times \mathcal{H}$. Direct maximization is challenging due to the presence of the unknown function $H$. 
Therefore,  a two-stage approach is used here.
We first estimate the the jump size of $H$ and then obtain the estimated transformation function. Then, we estimate the parameter $\theta$ by maximizing the pseudo likelihood, where $H$ is replaced by its estimated value.

\subsection{Estimation of the transformation function}
\label{subsec:estH}

The estimation procedure for $H$ follows the approach outlined by \cite{Deresa2021} and \citet{Tong2019}.
To apply the method of \citet{Tong2019}, which only considers data without censoring, we define $V = \min(T, C)$ and the corresponding censoring indicator $\zeta =\Delta + \xi $. Since $T$ and $C$ are conditionally independent of $A$ given $X$ and $W$, $V$ and $A$ are still independent. 
Thus, as \cite{Deresa2021} claimed, the dependent censoring model can be rewritten as an independent model.
We note that the distribution function of $V$ given $X$ and $W$ has the form
$$
 \begin{aligned}
 F_V(v;\theta,H)
 &= P(V<z|X,W) \\
&=P(T<z|X,W)+P(C<z|X,W)-P(T<z,C<z|X,W) \\
&=F_{\varepsilon _T}(H(z)-X^T \beta ) + F_{\varepsilon _C}(H(z)-W^T \eta )-
\mathcal{C} (F_{\varepsilon _T}(H(z)-X^T \beta ),F_{\varepsilon _C}(H(z)-W^T \eta )).
 \end{aligned}
$$
Thus the  density function of $V$ is 
$$
 \begin{aligned}
f_V(z|X,W)
&=f_{\varepsilon _T}(H(z)-X^T \beta )h(z) + f_{\varepsilon _C}(H(z)-W^T \eta )h(z) -
 \, \mathcal{C}'_1(F_{\varepsilon _T}(H(z)-X^T \beta ),F_{\varepsilon _C}(H(z)-W^T \eta ))
 f_{\varepsilon _T}(H(z)-X^T \beta )h(z)
 \\
& \quad\quad 
-\mathcal{C}'_2(F_{\varepsilon _T}(H(z)-X^T \beta ),F_{\varepsilon _C}(H(z)-W^T \eta ))f_{\varepsilon _C}(H(z)-W^T \eta )h(z),\\
 \end{aligned}
$$
where $h(z)=dH(z)/dz$.
Since $Z=min(V,A)$, where $V$ and $A$ are independent,
we have the following log-likelihood function:
\begin{equation}
    \begin{aligned}
    &\sum_{i=1}^{n} \zeta_i \left[ \log \omega \left( H(Z_i) - X_i^T \beta ,\, H(Z_i) - W_i^T \eta \right) 
    + \log H\{Z_i\} + \log n \right] \\
    & \qquad + \sum_{i=1}^{n} (1 - \zeta_i) \log S \left( H(Z_i) - X_i^T \beta ,\, H(Z_i) - W_i^T \eta \right) ,
\end{aligned}
\label{loglikefunc}
\end{equation}
where $\omega (x_1,x_2)=f_{\varepsilon_T }(x_1)(1-\mathcal{C}'_1(F_{\varepsilon_T }(x_1),F_{\varepsilon_C }(x_2)))+f_{\varepsilon_C }(x_2)(1-\mathcal{C}'_2(F_{\varepsilon_T }(x_1),F_{\varepsilon_C }(x_2)))$,  
$ S(x_1,x_2)= \tilde{C} \left( F_{\varepsilon_T }(x_1),  F_{\varepsilon_C }(x_2) \right)$.
Here, $H\{Z_i\} = H(Z_i) - H(Z_i -)$ represents the jump size of  $H$ at the point $Y_i$.
By maximizing the log-likelihood function (\ref{loglikefunc}) w.r.t. $H\{Z_i\}$, we can obtain the estimator $\hat{H}(\cdot )$.
Then  $\hat{\theta}$ can be derived by the solution of the score function
$$
\sum_{i=1}^{n} U(Z_i,\Delta_i,\xi_i,\theta,\hat{H}(\cdot ))=0,
$$
where $ U(Z_i,\Delta_i,\xi_i,\theta,H)=\partial \log L_i(\theta,H) /\partial \theta $.

\subsection{Estimation Algorithm  }

    Let $Z_1<Z_2<...<Z_n $. For the determination of $(\hat{\theta}_n,\hat{H}_n)$, we propose to employ the following algorithm.
 
    Step 1. Choose the initial values of $\theta^{(0)}$ and $H\{Z_k\}^{(0)}$, $k=1,...,n$.
    
    Step 2. Iteration of $t \to (t+1)$ . 
    
    (a) For $Z_k>0$, obain the updated estimator $H\{ Z_k \} ^{(t+1)}$ by
    \begin{equation}
        \begin{aligned}
    H\{Z_k\}^{(t+1)}
    =  \frac{\sum_{i=1}^{n} \zeta_i  I(Z_i=Z_k) }
    {\sum_{i=1}^{n}I(Z_i \ge Z_k)
    \psi^{(t)}(Z_i)    },
    \end{aligned}
    \label{eq:stepform}
    \end{equation}
    where 
    \begin{equation}
        \begin{aligned}
    \psi^{(t)}(Z_i) 
    &= -
    [\zeta_i \frac{\partial }{\partial H(z)} \log 
    \omega\left(H(z) - X_i^T \beta^{(t)}, \, H(z) - W_i^T \eta^{(t)} \right) \\
    & \qquad \left .  +(1-\zeta_i ) \frac{\partial }{\partial H(z)} \log S
    \left( H(z)-X_i^T\beta^{(t)} ,\, H(z)-W_i^T\eta^{(t)} \right) ]
    \right| _{H(z) = H^{(t)}(Z_i)},
    \end{aligned}
    \label{eq:alg_psi}
    \end{equation}
    for 
    $H(Z_i)=\sum_{k=1}^{n}H\{Z_k\}I(Z_k \le Z_i)I(Z_k \ge 0)$.

    (b) For $Z_k < 0$, obtain the updated estimator $H\{ Z_k \} ^{(t+1)}$ by
    $$
    H\{Z_k\}^{(t+1)}
    =\frac{\sum_{i=1}^{n} \zeta_i I(Z_i=Z_k) }
    {\sum_{i=1}^{n}I(Z_i \le Z_k)
    \psi^{(t)}(Z_i)   },
    $$
    where 
    $H(Z_i)=\sum_{k=1}^{n}H\{Z_k\}I(Z_k \ge Z_i)I(Z_k \le 0)$.

    Step 3. Obtain the estimator $\theta^{(t+1)}$ by solving the following score equation 
    $$
    U(\theta)=
    \sum_{i=1}^{n} U(Z_i,\Delta_i,\xi_i,\theta,H^{(t+1)}(Z_i))=0.
    $$    
    
    Step 4. Repeat Step 2 and Step 3 until convergence criteria are satisfied.

\subsection{Goodness of fit test}
\label{sec:goodness}

To assess the quality of the model fit, we employ a formal goodness-of-fit test, following the approach of \citet{Deresa2021}. This test measures the $L_2$ distance between the distribution derived from a nonparametric Kaplan-Meier estimator and that of the proposed semiparametric model.
We continue to use $V$ and $\zeta$ as defined in section \ref{subsec:estH} and consider the null hypothesis
$$
H_0 : P(V \le v)=F_{V}(v;\theta,H), 
\text{for some $\theta$}.
$$
According to the proposed model, the distribution of $V$ satisfies the following equation
$$
\begin{aligned}
F_V(v;\theta,H)
&=\int \int P(V \le v|X=x,W=w)f_{X,W}(x,w)dxdw\\
&=\int F_{\varepsilon_T} (H(v)-x^T\beta)f_X(x)dx + \int F_{\varepsilon_C} (H(v)-w^T\eta)f_W(w)dw \\
& \quad\quad -\int \int \mathcal{C} (F_{\varepsilon_T}(H(v)-x^T\beta),F_{\varepsilon_C}(H(v)-w^T\eta))f_{X,W}(x,w)dxdw.
\end{aligned}
$$
By substituting the estimated model parameters, we obtain the approximate distribution function as
\begin{equation} 
\begin{aligned}
\hat{F}_V(v;\hat{\theta}_n,\hat{H}_n)
&=\frac{1}{n}\sum_{i=1}^{n}  F_{\varepsilon_T,\hat{\lambda}_{T,n}} (\hat{H}_n(v)-X_i^T\hat{\beta}_n) + \frac{1}{n}\sum_{i=1}^{n}  F_{\varepsilon_C,\hat{\lambda}_{C,n}} (\hat{H}_n(v)-W_i^T\hat{\eta}_n) \\
& \quad \quad -\frac{1}{n}\sum_{i=1}^{n}  \mathcal{C}_{\hat{r}_n}(F_{\varepsilon_T,\hat{\lambda}_{T,n}}(\hat{H}_n(v)-X_i^T\hat{\beta}_n),F_{\varepsilon_C,\hat{\lambda}_{C,n}}(\hat{H}_n(v)-W_i^T\hat{\eta}_n)).
\end{aligned}
\end{equation}
Consider the nonparametric estimator of the distribution of $V$, which is based on the  estimator proposed by \citet{Kaplan1958} for right-censored survival data:
$$
\hat{F}_{KM}(z) = 1 - \prod_{Z_i \le z} \left\{ 1 - \frac{\sum_{j=1}^{n} I(Z_j = Z_i, \zeta_i = 1)}{\sum_{j=1}^{n} I(Z_j \ge Z_i)} \right\}.
$$
To test whether $H_0$ holds, we will assess the goodness-of-fit using the Cramér-von Mises statistic, which is computed as follows:
$$
T_{CM}=\sum_{i=1}^{n} \{\hat{F}_{KM}(Z_i) -  \hat{F}_V( Z_i; \hat{\theta}_n,\hat{H}_n) \}^2.
$$
A larger $T_{CM}$ value suggests that the proposed model may be inaccurate.
As mentioned by \cite{Deresa2021}, a bootstrap procedure is recommended to obtain the distribution of $T_{CM}$ under the null hypothesis. We will briefly outline the algorithm for the goodness-of-fit test for the proposed model as follows. Let $B$ denote the number of bootstrap samples.

1. 
We first generate the error terms $(\varepsilon_{T,i}^b, \varepsilon_{C,i}^b)$ from the fitted copula model using the inverse functions $\varepsilon _{T,i}^b=F_{\varepsilon_T,\hat{\lambda}_{T,n} }^{-1}(U_{T,i}^b )$ and $\varepsilon _{C,i}^b=F_{\varepsilon_C,\hat{\lambda}_{C,n} }^{-1}(U_{C,i}^b )$, where $(U_{T,i}^b, U_{C,i}^b)$ for $i=1,...,n$ are generated from the copula $\mathcal{C}_{\hat{r}_n}$.

2. We obtain  $(T_i^b,C_i^b), i=1,...,n $,  by 
$$
\left\{\begin{matrix} 
  T_i^b=\hat{H}_n^{-1}({X_i}^T \hat{\beta}_n + \varepsilon _{T,i}^b  ), \\  
  C_i^b=\hat{H}_n^{-1}({W_i}^T \hat{\eta}_n + \varepsilon _{C,i}^b  ).
\end{matrix}\right. 
$$

3. 
Obtain $A_i^b$ using the Kaplan-Meier estimator of the distribution of $A$, with $A$ treated as the survival time and $1-\zeta$ as the censoring indicator.

4.
Generate the bootstrap observations $(Z_i^b, \Delta_i^b, \xi_i^b, \zeta_i^b)$, $i=1, \ldots, n$, according to the definition of censoring in survival data, i.e., 
$$Z=min(T,C,A),
\quad \Delta=I(Z=T), 
\quad \xi=I(Z=C),
\quad  \text{and} 
\quad  \zeta=\Delta+\xi.$$

5. 
Calculate the bootstrap Cramér-von Mises statistic $T_{CM}^{*b}$ for each  bootstrap sample obtained in Step 4.

6. 
Repeat Steps 1-5 for each $b = 1, \ldots, B$. Then, we can estimate the $p$-value by
$$\hat{p}=\frac{1}{B} \sum_{i=1}^{B}
I( T_{CM}^{*b} > T_{CM}  ). $$

%% file: sections/Sec4_properties.tex
\section{Theoretical properties of the estimators}\label{sec:property}

The following regularity conditions are required to establish the asymptotic properties of our estimators.

\begin{enumerate}[(C1)]\setlength{\itemsep}{-4pt}
    \item\label{C1}   
    The true value $\theta_0$ lies within a known compact set $\Theta$ in $\mathcal{R}^{p+q+3}$.

    \item\label{C2}   The true function $H_0$ is monotonically  increasing and has a continuous and positive derivative.
    \item\label{C3}    (i) The covariates $X$ and $W$ are uniformly bounded, taking values in known compact sets $\mathcal{X} \subset \mathcal{R}^p$ and $\mathcal{W} \subset \mathcal{R}^q$, respectively.
	(ii) It is assumed that $E\{XX^T\}>0$ and $E\{WW^T\}>0$.
 
    \item\label{C4}  For any $\delta> 0$, there exists $\varepsilon  > 0$ such that 
    $$\inf_{ ||\theta - \theta_0|| >\delta} ||E[U\{ Z, \Delta, \xi, \theta, H_0(\cdot )\}]||> \varepsilon, $$ 
    where $||\cdot||$ is the Euclidean norm.
    \item\label{C5}  
    The density function $f$ is positive and satisfies
    $$\limsup_{|t| \to \infty} |t|^{1+v} f(t) < \infty,$$
    for some constant $v > 0$.

    \item\label{C6} 
    Consider the function $\varpi(z) = f_{\varepsilon_T}(z) (1 - \mathcal{C}'_1(F_{\varepsilon_T}(z), F_{\varepsilon_C}(z))) + f_{\varepsilon_C}(z) (1 - \mathcal{C}'_2(F_{\varepsilon_T}(z), F_{\varepsilon_C}(z)))$. 
     Define the function $\phi(z)=-d log(\varpi(z) )/dz $. The function $\phi(z)$ is strictly increasing and converges to finite limits as $z \to \pm \infty$, i.e.,
    \begin{equation}
        \lim_{z \to -\infty} \phi(z) = \phi_{-\infty}, 
    \quad
    \lim_{z \to +\infty} \phi(z) = \phi _{+\infty}.
    \label{eq:C6_limit}
    \end{equation}

    \item\label{C7}  
    The second derivative of the function $\phi$, defined in (C\ref{C6}), is bounded and continuous.
\end{enumerate}

It is worth noting that conditions (C\ref{C1})-(C\ref{C3}) are prevalent in the statistical literature. 
Condition (C\ref{C4}) is an identification condition introduced in \cite{Deresa2021}.
Condition (C\ref{C5}) establishes that the tail of density $f$ is not too heavy.
(C\ref{C5})-(C\ref{C7}) also used in \cite{Tong2019}.
Condition (C\ref{C6}) ensures that the estimator $\hat{H}(z)$ behaves appropriately at the boundaries, indicating that the limits of the denominator in \eqref{eq:stepform} exist and are finite as $z \to -\infty$ and $z \to \infty$. Condition (C\ref{C7}) guarantees the smoothness and stability of the model. Both Conditions (C\ref{C6}) and (C\ref{C7}) are crucial in the proof of Lemma \ref{lemma:distance}, which is subsequently used to establish the asymptotic normality of the proposed estimators.
 In addition, many density functions
satisfy conditions (C\ref{C4})–(C\ref{C7}) and one example is when the hazard function of $\varepsilon $
takes the form $exp(t)/(1 + \lambda exp(t))$ with $\lambda \ge 0$ \citep{Chen2002, Chentong}.
The proof that limits \eqref{eq:C6_limit} in Condition (C\ref{C6}) holds under this density is provided in the Supplementary Material, where it is demonstrated under three commonly used copula models: Frank, Gumbel, and Gaussian.

Subsequently, we present the asymptotic properties of the estimators $\hat{\theta}_n$ and $\hat{H}_n$. The proofs are given in the appendix.

\begin{theorem}
    (Consistency)
	Assume that Conditions (C\ref{C1})-(C\ref{C5}) given above hold true and that Assumptions (A\ref{A1})-(A\ref{A9}) are satisfied.
    Then for any $\tau \in (0,\infty)$, we have
     $\hat{\theta}_n \to \theta_0$  
     and
     $\sup_{t\in [-\tau,\tau]} |\hat{H}(t)- H_0(t)| \to 0$ almost surely.
\label{theorem:consistency}
\end{theorem}

\begin{theorem}
	 (Asymptotic normality)
    Assume that Conditions (C\ref{C1})-(C\ref{C7}) given above hold and that Asssumptions (A\ref{A1})-(A\ref{A9}) are satisfied. Then  as $n \to \infty$, we have that
    $\sqrt{n}(\hat{\theta}_n - \theta_0) \xrightarrow{d} N(0, A_\theta^{-1} \Sigma_\theta (A_\theta^{-1})^T),$
    where $A_\theta$ and $\Sigma_\theta$ are given by (\ref{Atheta}) and (\ref{Sigmatheta}) in the \ref{sec:mynorm2}, respectively.
\label{theorem:normal}
\end{theorem}

%% file: sections/Sec5_simulation.tex
\section{Simulation}\label{sec:simulation}

In this section, we assess the estimation procedure through an extensive simulation study. We compare our estimator with two existing approaches: the method presented in \cite{Deresa2022}, which accommodates a bivariate non-normal distribution of errors but requires a parametric specification of the transformation function, and the semiparametric method proposed by \cite{Deresa2021}, similar to ours, but assumes the error vector follows a standard bivariate normal distribution. The comparison with simulation results from the method of \cite{Deresa2021}  illustrates the effects of misspecifying the marginal distribution. Additionally, we conduct experiments to investigate the impact of misspecifying the copula structure on the results. Finally, we evaluate the performance of the proposed goodness-of-fit test by analyzing various simulation scenarios.

For the simulation, we consider the following copula models:
\begin{equation}
    \mathcal{C}_{r}(u,v)= -\frac{1}{r}
\log\{ 1 + \frac{(e^{-ru}-1 )(e^{-rv}-1 ) }{e^{-r}-1 }  \}, r\ne 0,
\label{simu:Frankcop}
\end{equation}
\begin{equation}
    \mathcal{C}_{r}(u,v)=  \Phi_r(\Phi^{-1}(u),\Phi^{-1}(v)), -1<r<1,
\label{simu:Claytoncop}
\end{equation}
\begin{equation}
    \mathcal{C}_{r}(u,v)= exp \{ - [ (-\log u)^r +  (-\log v)^r ]^{1/r} \}, r\ge 1,
\label{simu:Gumbelcop}
\end{equation}
which are usually referred to as Frank, Gaussian, and  Gumbel models respectively.
The range of the association parameter $r$ is different in the three copula families. In the following, we use Kendall’s $\tau$ as a global measure of association between the failure time $T$ and the observation time $C$. 
For the Frank copula \eqref{simu:Frankcop} we have an integral form, 
$\tau=1-\frac{1}{r}+\frac{1}{r^2}\int_{0}^{r} \frac{t}{e^t-1}dt $, 
while for Gaussian copula \eqref{simu:Claytoncop} and Gumbel copula \eqref{simu:Gumbelcop}, we have $\tau=2arcsin(r)/\pi$ and 
 $\tau=1-1/r$, respectively.

\textbf{Scenario 1}:
In this setting, we generate $i.i.d.$ data from the following  model:
$$
\left\{\begin{matrix} 
  H(T)=\beta_1 X_1 + \beta_2 X_2 + \varepsilon_T, \\  
  H(C)=\eta_1 X_1 + \eta_2 X_2 + \varepsilon_C, 
\end{matrix}\right. 
$$
where $X_1 \sim Ber(0.5)$, $X_2 \sim U_n(0, 5)$ with $X_1$ and $X_2$ being independent.
The vector
$(\varepsilon_T,\varepsilon_C)$ follows the distribution  $P(\varepsilon_T<t,\varepsilon_C<c)=\mathcal{C}_r (F_{\varepsilon_T}(t),F_{\varepsilon_C}(c))$, where the copula $\mathcal{C}_r $ is a Frank copula with $r=10$, which corresponds to Kendall's correlation $\tau=0.67$. 
The marginal distributions of $\varepsilon_T$ and $\varepsilon_C$ accord  to the hazard function
$exp(t)/(1 + \lambda exp(t))$ \citep{Chen2002,Chentong} with  
$\lambda_T=0.5$ and $\lambda_C=0.8$,  respectively.
The  regression parameters are $\beta_1=0.6$,  $\beta_2=1.4$, $\eta_1=0.4$, $\eta_2=0.8$.
The administrative censoring times are drawn from a uniform distribution over [0, 3].

We consider the following two different transformation functions and then compare our proposed method with the method proposed by \cite{Deresa2021} for estimation, referred to as the PT (Parametric Transformation) method 
$$
Case \, 1: H_1(z)= H_{Y}^{\alpha}(z), \text{where} \, \, \alpha=0.5, 
\qquad \qquad
Case  \, 2: H_2(z)=z^3 .
$$
where $ H_{Y}^{\alpha}(z)$ is the Yeo-Johnson family of transformations \citep{Yeo2000}, given by
$$
 H_{Y}^{\alpha}(z)= \begin{cases} 
[(z+1)^{\alpha } -1]/\alpha  ,\quad &z\ge 0, \alpha \ne 0 \\
log(z+1) ,\quad &z\ge 0, \alpha = 0 \\
-[(-z+1)^{2-\alpha }-1]/(2-\alpha ),\quad &z< 0, \alpha \ne 2 \\
-log(-z+1),\quad &z< 0, \alpha= 2 .
\end{cases} 
$$
with parameter $0 \le \alpha \le 2$. 
Since the PT method, which is based on parametric marginals and parametric copulas, is a  fully parametric approach,  
we implement their method assuming a parametric form for the transformation function. 
Specifically, we estimate their model using the Yeo-Johnson family of transformations \citep{Yeo2000}.
Since the first transformation function $H_1(z)$ is a special case of the \cite{Yeo2000} transformations family with a transformation parameter of 0.5, while $H_2(z)$  is not.
As a result, the model proposed by \cite{Deresa2022} is misspecified under Case 2.
Under this scenario, the average proportion of dependent censoring and administrative censoring are 35$\%$ and 15$\%$ respectively.

Table \ref{simu:compare}  presents the bias (BIAS), standard deviation (SD), root mean squared error of (RMSE) and  the $95\%$ empirical  coverage probability (CP) for two cases.
To evaluate bias, a commonly used rule of thumb suggests that biases are unlikely to  significantly impact  on inferences unless the standardized bias (bias as a percentage of the standard deviation) exceeds $40\%$ \citep{Olsen2001, Deresa2021}. 
The parameter estimators from  the proposed method and the \citet{Deresa2022} method are both below the threshold for the $H_1$ transformation, indicating that these methods produce satisfactory estimators for the model parameters.
It is expected that the  \citet{Deresa2022}  method outperforms the proposed method in terms of bias  under the Case 1, since the transformation function is correctly specified. 
However, the proposed method exhibits similar root mean squared errors and coverage probabilities to the method by \citet{Deresa2022} for the $H_1$ transformation. 
These CP values are close to the $95\%$ nominal level, indicating that  the asymptotic normality of our estimators is reasonable.
When the true transformation function is $H_2$, some parameter estimators from the \citet{Deresa2022} model are inconsistent, whereas the proposed method provides significantly better results with smaller bias and root mean squared error values. This highlights the flexibility of the proposed method. Additionally, a simulation evaluating the effects of $\lambda_T$ and $\lambda_C$ values is provided in the Supplementary Material. The results indicate that the proposed method performs well regardless of whether $\lambda_T = \lambda_C$ or not.

\begin{table}[]
\centering
\resizebox{0.8\linewidth}{!}{  
\begin{tabular}{lccccccccc}
\toprule
   &            & \multicolumn{4}{c}{Proposed method} & \multicolumn{4}{c}{PT method}  \\
\cmidrule(lr){2-6}\cmidrule(lr){7-10}
      & Parameters  & Bias     & SD     & RMSE   & CP     & Bias   & SD    & RMSE  & CP    \\
\midrule
$H_1$ & $\beta_1$   & 0.008    & 0.160  & 0.160  & 0.950  & -0.009 & 0.153 & 0.153 & 0.938 \\
      & $\beta_2$   & 0.010    & 0.241  & 0.241  & 0.966  & -0.005 & 0.162 & 0.162 & 0.950 \\
      & $\eta_1$    & 0.009    & 0.134  & 0.134  & 0.944  & -0.002 & 0.132 & 0.132 & 0.936 \\
      & $\eta_2$    & 0.006    & 0.174  & 0.174  & 0.950  & 0.020  & 0.167 & 0.168 & 0.942 \\
      & $\tau$      & 0.021    & 0.245  & 0.246  & 0.968  & 0.029  & 0.251 & 0.253 & 0.970 \\
      & $\lambda_T$ & 0.005    & 0.227  & 0.226  & 0.960  & 0.000  & 0.111 & 0.111 & 0.964 \\
      & $\lambda_C$ & 0.007    & 0.355  & 0.355  & 0.962  & -0.013 & 0.196 & 0.197 & 0.956 \\
\\
$H_2$ & $\beta_1$   & 0.004    & 0.157  & 0.157  & 0.948  & -0.104 & 0.140 & 0.174 & 0.874 \\
      & $\beta_2$   & -0.012   & 0.239  & 0.239  & 0.964  & -0.364 & 0.194 & 0.413 & 0.701 \\
      & $\eta_1$    & 0.008    & 0.131  & 0.131  & 0.948  & -0.026 & 0.123 & 0.125 & 0.946 \\
      & $\eta_2$    & 0.006    & 0.170  & 0.170  & 0.952  & -0.357 & 0.156 & 0.390 & 0.689 \\
      & $\tau$      & 0.026    & 0.268  & 0.269  & 0.969  & -0.066 & 0.275 & 0.283 & 0.970 \\
      & $\lambda_T$ & -0.006   & 0.216  & 0.216  & 0.950  & -0.374 & 0.110 & 0.385 & 0.524 \\
      & $\lambda_C$ & -0.014   & 0.315  & 0.315  & 0.960  & -0.478 & 0.158 & 0.503 & 0.654 \\
\bottomrule
\end{tabular}}
\caption{
\textcolor{red}{Simulation results based on 500 simulation runs  using the proposed method, \cite{Deresa2022} method (PT method)  under the transformation
functions $H_1$ and $H_2$.
The standard deviation (SD), root mean squared error (RMSE), and $95\%$ empirical  coverage probability (CP) are denoted by the respective abbreviations.}
 }  
\label{simu:compare}
\end{table}

\textbf{Scenario 2}:
We aim to evaluate the sensitivity of the proposed method to the misspecification of the error distribution. As stated, the proposed model is designed for nonnormal transformation models. 
In this scenario, we generate data from a Frank copula, with $T$ and $C$ following a $t$-distribution.
The remaining settings are the same as in Scenario 1 under Case 1.
We then  compare the results of the proposed method with those of \cite{Deresa2021}, which employs the bivariate normal transformation model for dependent censoring (referred to as the NT method).
Consequently, both the proposed method and the NT method are misspecified. It is straightforward to demonstrate that as the degrees of freedom of the $t$-distribution increase, the distribution approaches normality. Therefore, we consider two degrees of freedom, $d = 2$ and $d = 8$, to represent strong and weak deviations from the normal model, respectively.

Table \ref{simu:com_NT} presents the bias, SD, RMSE, and CP for the two methods. For the NT method, the values are smaller when $d = 8$ than when $d = 2$, consistent with the findings of \cite{Deresa2021}, indicating that the NT method is more suitable for normal error terms. For the proposed method, the values for $d = 2$ are significantly smaller than those for $d = 8$ and are similar to the results for the proposed method in Table \ref{simu:compare}, suggesting that our method is robust in this case. As the distribution approaches a normal transformation, the bias and RMSE increase and exceed those reported from NT method, indicating that the proposed method is not suitable when the error terms are approximately normal. Therefore, we recommend the goodness-of-fit test presented in Section \ref{sec:goodness}.

\begin{table}[]
\centering
\resizebox{0.8\linewidth}{!}{  
\begin{tabular}{lccccccccc}
\toprule
&            & \multicolumn{4}{c}{Proposed method} & \multicolumn{4}{c}{NT method}  \\
\cmidrule(lr){2-6}\cmidrule(lr){7-10}
      & Parameters & Bias     & SD     & RMSE   & CP     & Bias   & SD    & RMSE  & CP    \\
\midrule
$d=2$ & $\beta_1$  & -0.001   & 0.133  & 0.133  & 0.942  & -0.220 & 0.130 & 0.256 & 0.608 \\
      & $\beta_2$  & 0.062    & 0.187  & 0.197  & 0.936  & -0.641 & 0.130 & 0.654 & 0.310 \\
      & $\eta_1$   & 0.078    & 0.124  & 0.147  & 0.920  & -0.063 & 0.130 & 0.144 & 0.930 \\
      & $\eta_2$   & 0.293    & 0.169  & 0.338  & 0.597  & -0.172 & 0.117 & 0.208 & 0.678 \\
      & $\tau$     & 0.178    & 0.419  & 0.455  & 0.998  & 0.235  & 0.058 & 0.242 & 0.413 \\
      \\
$d=8$ & $\beta_1$  & 0.356    & 0.202  & 0.409  & 0.591  & -0.093 & 0.127 & 0.149 & 0.901 \\
      & $\beta_2$  & 0.812    & 0.325  & 0.875  & 0.311  & -0.192 & 0.227 & 0.359 & 0.687 \\
      & $\eta_1$   & 0.233    & 0.143  & 0.273  & 0.629  & 0.035  & 0.118 & 0.123 & 0.941 \\
      & $\eta_2$   & 0.457    & 0.200  & 0.499  & 0.371  & 0.051  & 0.133 & 0.143 & 0.931 \\
      & $\tau$     & 0.011    & 0.267  & 0.267  & 0.998  & 0.193  & 0.095 & 0.229 & 0.409 \\
\bottomrule
\end{tabular}}
\caption{
Simulation results based on 500 simulation runs  using the proposed method, \cite{Deresa2021} method (NT method)  when $T$ and $C$ follow $t$-distribution.
The standard deviation (SD), root mean squared error (RMSE), and $95\%$ empirical  coverage probability (CP) are denoted by the respective abbreviations.}
\label{simu:com_NT}
\end{table}


\textbf{Scenario 3}:
In this scenario, we are interested in the sensitivity of model parameter estimation to the misspecification of the copula structure. To investigate this, we generate data similar to Scenario 1 for $H_1$. However, we analyze it using Gumbel and Gaussian copula models while keeping the marginal distribution for the error terms vector correct. Since the data is generated from a model with a Frank copula, our analysis is based on a misspecified copula. Table \ref{simu:copula} shows the bias, SD, RMSE, and CP of the parameter estimators.

The results indicate that most of the estimators remain unbiased, and the CP values are close to the nominal level, suggesting that the estimation of the regression parameters is not significantly affected by the misspecified copula structure. This outcome is consistent with the findings of \cite{Huang2008}, \cite{Chen2010}, and \cite{Deresa2022}, who observed that copula structure misspecification typically results in moderate bias. 
We also observe that the bias of $\lambda_T$ and $\lambda_C$ under the Gaussian copula  are larger than under the Gumbel copula, likely due to  differences in the shapes of copula functions.
Additionally, Kendall's $\tau = 0.67$ corresponds to the Gaussian association parameter $r = 0.8653$, which is close to the bounds of $r$, potentially leading to bias in the marginal distribution estimators.

\begin{table}[!htb]
\centering
\begin{tabular}{ccccccccc}
\toprule
        & \multicolumn{4}{c}{Gumbel copula} & \multicolumn{4}{c}{Gaussian copula} \\
\cmidrule(lr){2-5}\cmidrule(lr){6-9}
            & Bias    & SD     & RMSE   & CP    & Bias     & SD     & RMSE   & CP     \\
\midrule
$\beta_1$    & 0.082   & 0.203  & 0.231  & 0.922 & 0.043    & 0.145  & 0.151  & 0.945  \\
$\beta_2$    & 0.098   & 0.249  & 0.297  & 0.924 & -0.092   & 0.235  & 0.255  & 0.916  \\
$\eta_1$     & 0.060   & 0.158  & 0.163  & 0.926 & 0.057    & 0.149  & 0.165  & 0.911  \\
$\eta_2$     & 0.020   & 0.190  & 0.191  & 0.948 & -0.053   & 0.154  & 0.163  & 0.924  \\
$\tau$      & -0.031  & 0.141  & 0.144  & 0.964 & -0.002   & 0.119  & 0.124  & 0.935  \\
$\lambda_T$ & 0.091   & 0.229  & 0.235  & 0.926 & -0.145   & 0.177  & 0.229  & 0.867  \\
$\lambda_C$ & 0.057   & 0.333  & 0.338  & 0.954 & -0.190   & 0.335  & 0.385  & 0.947 \\
\bottomrule
\end{tabular}
\caption{Simulation results  based on 500 simulation runs  using the proposed method, with Gumbel and Gaussian copulas  when the data are generated from a Frank copula.
The standard deviation (SD), root mean squared error (RMSE), and empirical $95\%$ coverage probability (CP) are denoted by the respective abbreviations.}
\label{simu:copula}
\end{table}


\textbf{Scenario 4}:
In this scenario, 
we aim to assess the effectiveness of the goodness-of-fit test detailed in Section \ref{sec:goodness} across three distinct cases.
For each case, we consider $B= 300$ bootstrap samples obtained from 500 datasets of size $n=300$ and $n=500$, and computed the Cramér-von Mises statistic.
For assessment of the power of the tests, 
we record the rejection rates at the $5\%$ and $10\%$  significance levels over 500 runs.
In the first case, we generate data from  the set same as Scenario 1 for the $H_1$ transformation and then estimate it with correct copula and marginal distribution.
Next, in the second case, we conduct the test with the simulation result of Scenario 2 under  $d=8$, which means we use the misspecified marginal distribution. 
In the third case, we analyzed the data with the misspecified copula function that discussed in the Scenario 3 under Gumbel copula.

Table \ref{simu:goodness} implies the goodness-of-fit test performance. Under Case 1, the rejection rates align closely with the nominal levels.
The mean value of the test statistic $T_{CM}$ closely approximates the mean value of the bootstrap statistics $B^{-1}\sum_{i=1}^{B}T_{CM}^{*b}$.
We observe a similar result for Case 3, aligning with the findings of Scenario 3, which suggests that the proposed method is relatively robust to copula structure misspecification.
The Case 2 in Table \ref{simu:goodness}  demonstrates that the proposed test is highly effective at rejecting the null hypothesis, when the marginal distributions are misspecified. This finding highlights the robustness and reliability of the test, which is consistent with the result of \cite{Deresa2021}.
Besides,   the power of the test increases as the sample size grows from $n = 300$ to $500$.

\begin{table}[!ht]
\centering
\begin{tabular}{cccccc}
\toprule
        & Sample size & $E(T_{CM})$ & $E(T_{CM}^{*,b})$ & Rejection rate at $5\%$ & Rejection rate at $10\%$ \\
\midrule
Case 1 & 300         & 0.116  & 0.108      & 0.041                 & 0.089                  \\
       & 500         & 0.239  & 0.227      & 0.047                 & 0.093                  \\
Case 2 & 300         & 0.407  & 0.283      & 0.586                 & 0.791                  \\
       & 500         & 0.819  & 0.409      & 0.783                 & 0.969                  \\
Case 3 & 300         & 0.134  & 0.132      & 0.054                 & 0.099                  \\
       & 500         & 0.241  & 0.232      & 0.058                 & 0.117                 \\
\bottomrule
\end{tabular}
\caption{Simulation results for evaluating the performance of the goodness-of-fit test.  }  
\label{simu:goodness}
\end{table}

%% file: sections/Sec6_real.tex
\section{Real data example}
\label{sec:real}

We illustrate the proposed method using an application to the BMT example, where data were collected on patients with acute leukaemia following allogeneic bone marrow transplantation.
The dataset is available in the R package KMsurv and a more detailed account can be found in \cite{Klein1997SurvivalAT}, page 464.
A key focus is the incidence of chronic GVHD, a common complication following transplantation. The GVHD endpoint was potentially subject to dependent censoring due to death. Among the 137 patients who received bone marrow transplants, 61 developed chronic GVHD.
During the follow-up period, 81 patients ($59\%$) passed away, with 52 ($38\%$) of them dying before the onset of chronic GVHD. Patients were categorized into three risk groups at the time of transplantation: ALL, AML low-risk, and AML high-risk. Let $T$ denote the time to the chronic GVHD endpoint, $C$ represent the time to death, and $A$ indicate the duration of the follow-up period, which is variable rather than fixed.

By considering the potential dependence between the time to chronic GVHD development and the time to death, \cite{Li2015} analyzed this dataset using a known Frank copula model.
Their findings suggest a moderate positive association between GVHD and death, given the disease type and age.
In our analysis, we apply the proposed estimation technique to fit the model (\ref{transmodel})-(\ref{eq23}) to this dataset.
We include the following nine variables as covariates in our analysis:
Group (1 for AML high-risk group, 0 for other groups),
patient age (in years),
donor age (in years),
patient gender (0 for female, 1 for male),
donor gender (0 for female, 1 for male),
patient CMV status (0 for CMV Negative, 1 for CMV Positive),
donor CMV status (0-CMV Negative and  1-CMV Positive),
waiting time to transplant (in days),
FAB (0 for Otherwise and 1 for FAB Grade 4 Or 5 and AML). 
We aim to examine the impact of AML high risk disease group on the development of GVHD and death after adjusting for covariates and to estimate the association between $T$ and $C$.


\begin{table}[!ht]
\centering
\begin{tabular}{lcccccc}
\toprule
                   & \multicolumn{3}{c}{Time to GVHD end point} & \multicolumn{3}{c}{Time to death} \\
\cmidrule(lr){2-4}\cmidrule(lr){5-7}
                   & Estimated      & SE         & $p$-value      & Estimated   & SE      & $p$-value   \\
\midrule
Group              & -0.506         & 0.223      & 0.023        & -0.521      & 0.426   & 0.221     \\
Patient age        & 0.240          & 0.116      & 0.038        & 0.198       & 0.113   & 0.081     \\
Donor age          & -0.344         & 0.118      & 0.004        & -0.333      & 0.208   & 0.110     \\
Patient gender     & 0.026          & 0.175      & 0.884        & 0.108       & 0.204   & 0.596     \\
Donor gender       & 0.231          & 0.160      & 0.149        & 0.237       & 0.186   & 0.203     \\
Patient CMV status & -0.494         & 0.139      & 0.000        & -0.230      & 0.179   & 0.199     \\
Donor CMV status   & 0.259          & 0.163      & 0.112        & 0.144       & 0.134   & 0.283     \\
Waiting time       & -0.205         & 0.061      & 0.001        & -0.211      & 0.067   & 0.002     \\
FAB                & 0.155          & 0.222      & 0.483        & 0.060       & 0.197   & 0.763     \\
$\tau$             & 0.661          & 0.323      & 0.041        &             &         &       \\    
\bottomrule
\end{tabular}
\caption{Parameter estimates, bootstrap standard errors and $p$-values using the proposed method for the BMT dataset.}
\label{Table:real}
\end{table}


In Table \ref{Table:real}, we present the parameter estimates, bootstrap standard errors, and $p$-values. The bootstrap standard errors are calculated using 500 bootstrap samples. The $p$-values are based on a Wald statistic, utilizing these bootstrap standard errors. The estimated association parameter for the two endpoints corresponds to a Kendall's correlation $\tau$ of 0.661, indicating a correlation between the time to GVHD and time to death, even after adjusting for covariates. This relatively strong correlation suggests that neglecting this dependency could introduce bias in the parameter estimates.
Furthermore, we find that five covariates (group, patient age, donor age, patient CMV status, and waiting time) have significant effects on the time to the GVHD endpoint, while one covariate (waiting time) significantly affects the time to death. The results in Table \ref{Table:real} highlight the importance of waiting time for both time to chronic GVHD and time to death, suggesting that a short waiting time may effectively delay the onset of both endpoints. 

Additionally, patients in the AML high-risk group exhibit significantly faster progression to chronic GVHD compared to those in the ALL group and AML low-risk group, as indicated by a negative estimated coefficient. However, there is no significant association between the AML high-risk group and time to death, as the $p$-value exceeds 0.05.

\begin{figure} [!htb] 
\centering
\subfigure[ The estimated error $\hat{\varepsilon}_T$.]{
\includegraphics[width=0.4\textwidth]{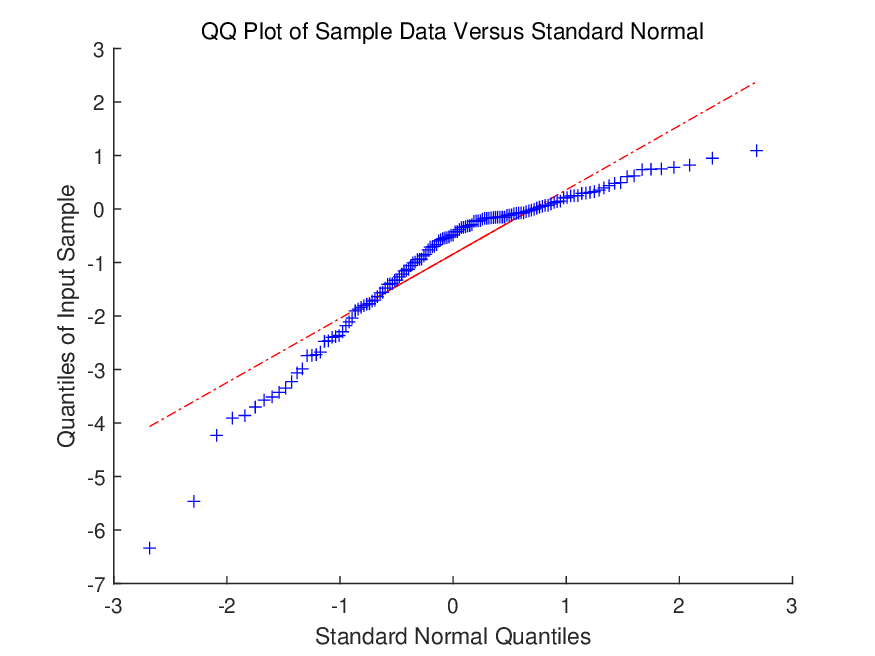}
\label{real:error1}}
\,\,\,\,\,
\subfigure[The estimated error $\hat{\varepsilon}_C$.]{
\includegraphics[width=0.4\textwidth]{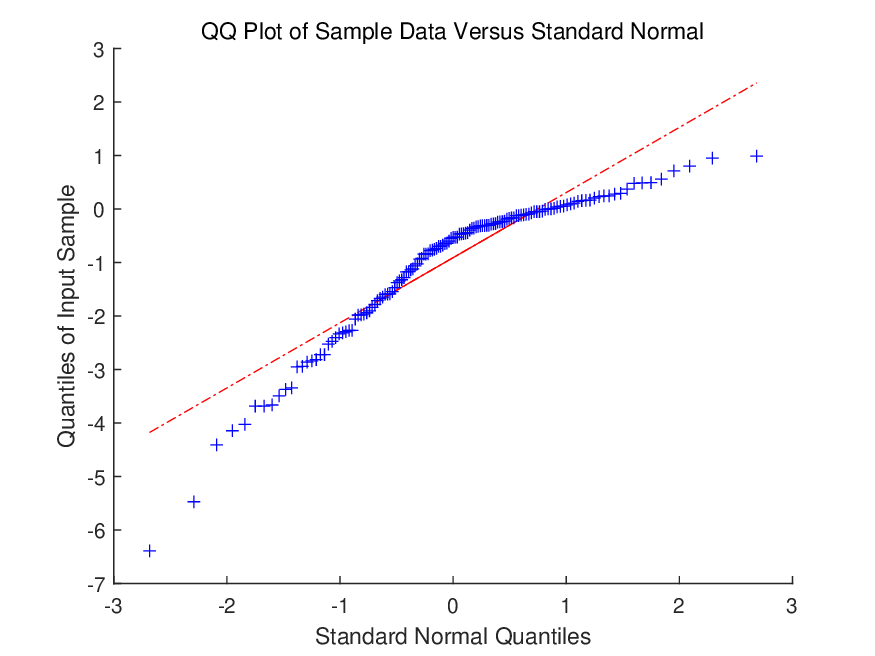}
\label{real:error2}}
\caption{ Quantile plot of the estimated errors against the standard normal distribution.}
\label{figure:error}
\vspace{-4mm}
\end{figure}

Since the proposed model (\ref{transmodel})-(\ref{eq23}) is based on nonnormal error distributions, we aim to test the normality of the data. We generated a Quantile-Quantile (QQ) plot of the errors to assess their normality, as shown in Figure \ref{figure:error}. It is evident that, for the BMT data analyzed in this study, the errors deviate from a normal distribution.
To further evaluate the fit of the proposed model, we conducted a goodness-of-fit test as outlined in Section \ref{sec:goodness}. The Cramér–von Mises type statistic $T_{CM}$ yielded a $p$-value of 0.392 based on 500 bootstrap replications, indicating that the proposed semiparametric model cannot be rejected.

%% file: sections/Sec7_conclusion.tex
\section{Concluding remarks and future research }
\label{sec:conclusion}

In this paper, we propose a transformed linear model to exploit the association between the survival time $T$ and the dependent censoring time $C$.
To the best of our knowledge, there is limited availability of semiparametric models and inference techniques in the literature for analyzing bivariate survival data with non-normal error distributions. A key advantage of our approach is that the association parameter of the copula function remains unspecified and can be identified even when we observe only either $T$ or $C$ but not both. This identification is achieved by making certain assumptions and conditions outlined in Section \ref{sec:model}. 
As an example, we consider the marginal function (\ref{eq_tong}) in our model, which includes the Cox model and the proportional odds model as special cases when the error follows the extreme value distribution and the logistic distribution, respectively.

Our model offers flexibility for extension in various ways.  Firstly,  the current model can be adapted  to incorporate semiparametric or nonparametric regression functions using techniques like splines, orthogonal series, or kernel methods. 
Secondly, the application of more comprehensive survival models, such as competing risks models, cure models, and truncation in combination with dependent censoring, holds promising potential. 
Thirdly, it would be interesting to investigate whether more efficient copula functions and marginal distributions can be obtained.

%% file: sections/AppendixA_identification.tex
\section{Proof of Theorem \ref{theorem:identification}}\label{sec:proof31}

Define a pair of parameters
$s_u=(\beta_u,\eta_u,\lambda_{T,u},\lambda_{C,u},r_u, H_u) \in S$ for $u=1,2$.
We will complete this proof in four steps, following the identification proof methodology outlined in \cite{Deresa2021} and \cite{Deresa2022}.

(Step 1) 
We first prove the identification of $\beta,\eta,H$. 
Since $P(Z = T) > 0$ under Assumption (A\ref{A4}), we can consider $\Delta = 1$ and $\xi= 0$. From the likelihood function it follows that

\begin{equation}\label{proof31:step1}
    \begin{aligned}
&f_{\varepsilon _T,\lambda_{T,1}}(H_1(z)-x^T\beta_1)
S_{\varepsilon _C|\varepsilon _T,\lambda_{T,1},\lambda_{C,1},r_1} (H_1(z)-w^T\eta_1 | H_1(z)-x^T\beta_1)
P(A>z) h_1(z)\\
&=
f_{\varepsilon _T,\lambda_{T,2}}(H_2(z)-x^T\beta_2)
S_{\varepsilon _C|\varepsilon _T,\lambda_{T,2},\lambda_{C,2},r_2} (H_2(z)-w^T\eta_2 | H_2(z)-x^T\beta_2)
P(A>z) h_2(z)
\end{aligned}
\end{equation}
for all $z, x, w$ .
We know from  (A\ref{A6}) that
$$
	\lim_{z \to -\infty} S_{\varepsilon _C|\varepsilon _T,\lambda_{T,j},\lambda_{C,j},r_j}
	(H_j(z)-w^T\eta_j | H_j(z)-x^T\beta_j)=1 ,\quad
	j=1,2.
$$
When $z$ approaches $-\infty$, by taking the limit on both sides of equation (\ref{proof31:step1}), we have
\begin{equation}
    \lim_{z \to -\infty}
\frac{f_{\varepsilon _T,\lambda_{T,1}}(H_1(z)-x^T\beta_1)h_1(z)}
{f_{\varepsilon _T,\lambda_{T,2}}(H_2(z)-x^T\beta_2)h_2(z)}  =1, 
\label{eq:ratiofh}
\end{equation}
for almost all $z$ and $x$.
From Assumption (A\ref{A8}), we know that there exist a function $k(z)$ that satisfies
\begin{equation}
    \lim_{z \to -\infty} \frac{h_1(z)k(z)}{h_2(z)k(z)}=a_0,
\label{eq:ratioh1h2}
\end{equation}
where $a_0>0$ is a constant.
Thus for \eqref{eq:ratiofh},  we multiply the numerator and denominator  by the function $k(z)$ and derive
\begin{equation}
    \lim_{z \to -\infty}
\frac{f_{\varepsilon _T,\lambda_{T,1}}(H_1(z)-x^T\beta_1)}
{f_{\varepsilon _T,\lambda_{T,2}}(H_2(z)-x^T\beta_2)}=a,
\label{proof:rate}
\end{equation}
where $a=1/a_0$.
By L'Hopital's rule,  (\ref{proof:rate})  leads to 
\begin{equation}
        \lim_{z \to -\infty}
\frac{f'_{\varepsilon _T,\lambda_{T,1}}(H_1(z)-x^T\beta_1)h_1(z)}
{f'_{\varepsilon _T,\lambda_{T,2}}(H_2(z)-x^T\beta_2)h_2(z)}=a.
\label{eq:ratio_derif}
\end{equation}
We again use the steps in \eqref{eq:ratiofh} - \eqref{proof:rate}. Then the limit equation (\ref{eq:ratio_derif})  leads to 
$$
    \lim_{z \to -\infty}
\frac{f'_{\varepsilon _T,\lambda_{T,1}}(H_1(z)-x^T\beta_1)}
{f'_{\varepsilon _T,\lambda_{T,2}}(H_2(z)-x^T\beta_2)}=a^2.
$$
Recall that $ g_{\varepsilon _T,\lambda_{T,1}}(z)=\partial \log f_{\varepsilon _T,\lambda_{T,1}}(z) /\partial z  $ and $g_{\varepsilon _T,\lambda_{T,2}}(z)=\partial \log f_{\varepsilon _T,\lambda_{T,2}}(z) /\partial z$.
Then we have
\begin{equation}
    \lim_{z \to -\infty} 
    \frac{g_{\varepsilon _T,\lambda_{T,1}}(H_1(z)-x^T\beta_1)}
    {g_{\varepsilon _T,\lambda_{T,2}}(H_2(z)-x^T\beta_2) }
    =\lim_{z \to -\infty} 
    \frac{f'_{\varepsilon _T,\lambda_{T,1}}(H_1(z)-x^T\beta_1)}
    {f_{\varepsilon _T,\lambda_{T,1}}(H_1(z)-x^T\beta_1) } 
    \frac{f_{\varepsilon _T,\lambda_{T,2}}(H_2(z)-x^T\beta_2)}
    {f'_{\varepsilon _T,\lambda_{T,2}}(H_2(z)-x^T\beta_2) }
    =a.
\end{equation}
Besides from L‘Hopital rule of limits and \eqref{eq:ratioh1h2}, we have
\begin{equation}
 \lim_{z \to -\infty}
\frac{H_1(z)}{H_2(z)}
=\lim_{z \to -\infty}
\frac{h_1(z)}{h_2(z)}
=\frac{1}{a}
\label{proof:h1h2}
\end{equation}
By Assumption (A\ref{A9}), 
\eqref{proof:rate} and  \eqref{proof:h1h2} imply that $a=1$.
Then from Lemma \ref{lemma_H1H2}, we have that 
$\lim_{z \to -\infty} \{ H_1(z) - H_2(z) \}=0$  and 
$\beta_1=\beta_2$. 
Applying the same proof for $\Delta=0$ and $\xi=1$, we find that  $\eta_1=\eta_2$.

(Step 2) 
Then we  identify the parameters  $\lambda_T$ and $\lambda_C$.
We represent the previously identified parameters $\beta$ and $\eta$ as $\beta = (\beta_1, \ldots, \beta_p)^T$ and $\eta = (\eta_1, \ldots, \eta_q)^T$. 
Consider the survival function of the observable random variable $Z$ given $X$ and $W$, which is defined as:
$$
\begin{aligned}
P(Z \ge z| X=x,W=w)
&=P(H(T) \ge H(z),H(C) \ge H(z)| X=x,W=w)P(A\ge z)\\
&=P(\varepsilon _T \ge H(z)-x^T\beta,\varepsilon _C \ge H(z)-w^T\eta)P(A\ge z)\\
&=P(A\ge z) \tilde{\mathcal{C}}  (F_{\varepsilon _T} ( H(z)-x^T\beta),F_{\varepsilon _C} (H(z)-w^T\eta)).\\
\end{aligned}
$$
Since $P(A \le z)$ is identifiable and $\tilde{\mathcal{C}}(u,v)=1-u-v + \mathcal{C}(u,v)$, we consider that
\begin{equation}
    \begin{aligned}
&F_{\varepsilon _T,\lambda_{T,1}} ( H_1(z)-x^T\beta)
+F_{\varepsilon _C,\lambda_{C,1}} (H_1(z)-w^T\eta) 
 -\mathcal{C} _{r_1} (F_{\varepsilon _T,\lambda_{T,1}} ( H_1(z)-x^T\beta),F_{\varepsilon _C,\lambda_{C,1}} (H_1(z)-w^T\eta))\\
&=F_{\varepsilon _T,\lambda_{T,2}} ( H_2(z)-x^T\beta)
+F_{\varepsilon _C,\lambda_{C,2}} (H_2(z)-w^T\eta) 
 - \mathcal{C} _{r_2} (F_{\varepsilon _T,\lambda_{T,2}} ( H_2(z)-x^T\beta),F_{\varepsilon _C,\lambda_{C,2}} (H_2(z)-w^T\eta)).
\end{aligned}
\label{proof1:func_C}
\end{equation}
We take the derivative with respect to $x$ on both sides of the equation and, after a simple transformation, we obtain
$$
\frac{f_{\varepsilon _T,\lambda_{T,1}} ( H_1(z)-x^T\beta) }
{f_{\varepsilon _T,\lambda_{T,2}} ( H_2(z)-x^T\beta)}
=
\frac{1 - \mathcal{C}^{(1)} _{r_2} (F_{\varepsilon _T,\lambda_{T,2}} ( H_2(z)-x^T\beta),F_{\varepsilon _C,\lambda_{C,2}} (H_2(z)-w^T\eta)))}
{1 - \mathcal{C}^{(1)} _{r_1} (F_{\varepsilon _T,\lambda_{T,1}} ( H_1(z)-x^T\beta),F_{\varepsilon _C,\lambda_{C,1}} (H_1(z)-w^T\eta))},
$$
where $\mathcal{C}^{(1)}_r(u,v)=d\mathcal{C}_r(u,v)/du$ and $\mathcal{C}_r^{(2)}(u,v)=d\mathcal{C}_r(u,v)/dv$. 
Notably, $w$ only appears on the right side. Since $\text{var}(W) > 0$, the left side equals a constant for all $z$ and $x$.
By setting $z=0$, we obtain $\lambda_{T,1} = \lambda_{T,2}$.
Repeating the same arguments for the derivatives with respect to $w$, we identify that $\lambda_{C,1} = \lambda_{C,2}$.

(Step 3) 
In this step, we  identify the parameters  $r$.
Let $\lambda_{T}$ and $\lambda_{C}$ represent the already identified parameters of  $f_{\varepsilon _T }$ and $f_{\varepsilon _C }$ respectively. 
Considering \eqref{proof1:func_C} and letting $z=0$, we have
$$
\begin{aligned}
\mathcal{C}_{r_1} (F_{\varepsilon _T,\lambda_{T}} ( -x^T\beta),F_{\varepsilon _C,\lambda_{C}} (-w^T\eta))
=\mathcal{C}_{r_2} (F_{\varepsilon _T,\lambda_{T}} ( -x^T\beta),F_{\varepsilon _C,\lambda_{C}} (-w^T\eta)), 
\end{aligned}
$$
for all $x$ and $w$.
Since the copula is unique, it yields that $r_1=r_2$.

(Step 4)  
Finally, we aim to prove the identification of the transformation function $H$.
From Step2, we know that
$$
\begin{aligned}
\tilde{\mathcal{C}}_{r} (F_{\varepsilon _T,\lambda_{T}} ( H_1(z)-x^T\beta),F_{\varepsilon _C,\lambda_{C}} (H_1(z)-w^T\eta))
= \tilde{\mathcal{C}}_{r} (F_{\varepsilon _T,\lambda_{T}} ( H_2(z)-x^T\beta),F_{\varepsilon _C,\lambda_{C}} (H_2(z)-w^T\eta)),
\end{aligned}
$$
for all $z,x,w$. 
This is only possible if $H_1\equiv H_2$.
This completes the proof.

$\hfill\blacksquare$

%% file: sections/AppendixB_consistency.tex
\section{Proof of Theorem \ref{theorem:consistency}}\label{sec:myconsis2}

Now we prove the consistency of $(\hat{\theta}_n,\hat{H}_n)$, where $\hat{\theta}_n=(\hat{\beta}_n,\hat{\eta}_n,\hat{\lambda}_{T_n},\hat{\lambda}_{C_n},\hat{r}_n) $ .
The following proof is adapted from Theorem 2.2 of \cite{Tong2019}.
We begin by outlining the proof of Theorem \ref{theorem:consistency}. 
It is sufficient to demonstrate that for any convergent subsequence $(\hat{\theta}_{n_k}, \hat{H}_{n_k})$ of
$(\hat{\theta}_n, \hat{H}_n)$ that satisfies $(\hat{\theta}_{n_k}, \hat{H}_{n_k}) \to (\theta^*,H^*) $, it holds that $(\theta^*,H^*)=(\theta_0,H_0)$. 
To simplify the proof, we use  $\{n_k,k \ge 1\}$ as the original sequence. 
Given that $(\hat{\theta}_n, \hat{H}_n)$ is a symmetric function of the sample $X_1,...,X_n,W_1,...,W_n$, the pair $(\theta^*,H^*)$   is measurable with respect to the exchangeable  $\sigma$-field of ${X_n,n \ge 1}$. 
The Hewitt–Savage’s 0-1 law
implies that $(\theta^*,H^*)$ is a constant. 
In this proof, we use Lemma \ref{lemma:distance}, which, along with its proof, is provided in the \ref{lemma:distance}.

(Step 1)
Consider the likelihood (\ref{loglikefunc}) and  let  $V=min(T,C)$, $\zeta = \Delta+ \xi$.
We assume that there are $m$ observations that denote survival and dependent censoring times.
We write 
$l_n(\theta,H)=l_n(\theta,H)^{(1)} + l_n(\theta,H)^{(2)} $, 
where
\begin{equation}
    \begin{aligned}
    &l_n(\theta,H)^{(1)}= \frac{1}{n}\sum_{i=1}^{n} \zeta_i \left[ \log \omega \left( H(Z_i) - X_i^\top \beta ,\, H(Z_i) - W_i^\top \eta \right) 
    + \log H\{Z_i\} + \log n \right], \\
    &l_n(\theta,H)^{(2)}= \frac{1}{n} \sum_{i=1}^{n} (1 - \zeta_i) \log \tilde{C} \left( H(Z_i) - X_i^\top \beta ,\, H(Z_i) - W_i^\top \eta \right) ,
\end{aligned}
\end{equation}
We first consider $l_n(\theta,H)^{(1)}$ and have 
\begin{equation}
    \begin{aligned}
l_n(\theta,H)
&= \frac{1}{n}
\sum_{i=1}^{n} \zeta_i  [\log
\omega ( H(Z_i)-X_i^T\beta , H(Z_i)-W_i^T\eta  )
+\log H\{Z_i\} + \log (n)] \\
&=:\frac{1}{n}  \sum_{i=1}^{n} \zeta_i l_{i}(\hat{\theta}_n,\hat{H}_n), \\
\end{aligned}
\end{equation}
where $\omega (x_1,x_2)=f_{\varepsilon_T }(x_1)S_{\varepsilon_C |\varepsilon_T }(F_{\varepsilon_C }(x_2)|F_{\varepsilon_T }(x_1)) + f_{\varepsilon_C }(x_2)S_{ \varepsilon_T| \varepsilon_C }(F_{\varepsilon_T }(x_1) | F_{\varepsilon_C }(x_2))$.
For the simplicity of notation, we let 
 $ \sum_{i=1}^{m}    l_{i}(\hat{\theta}_n,\hat{H}_n)  $ 
 denote $\sum_{i=1}^{n}  \zeta_i  l_{i}(\hat{\theta}_n,\hat{H}_n) $.
Then we have 
\begin{equation} 
    \begin{aligned}
    l_n(\theta,H)
&=\frac{m}{n} \frac{1}{m} \sum_{i=1}^{m}I_{\{Z_i\le t_0\}}   l_{i}(\hat{\theta}_n,\hat{H}_n) 
 +\frac{m}{n} \frac{1}{m}\sum_{i=1}^{m}I_{\{t_0< Z_i\le t_N\}}   l_{i}(\hat{\theta}_n,\hat{H}_n)
 +\frac{m}{n} \frac{1}{m}\sum_{i=1}^{m}I_{\{Z_i > t_N\}}  l_{i}(\hat{\theta}_n,\hat{H}_n) \\
 &=:\frac{m}{n} \kappa_{1}(\hat{\theta}_n,\hat{H}_n) 
+\frac{m}{n}  \kappa_{2}(\hat{\theta}_n,\hat{H}_n) 
+\frac{m}{n}  \kappa_{3}(\hat{\theta}_n,\hat{H}_n),
\end{aligned}
\end{equation}
where $-M=t_0<t_1<...<t_N=M$ are fixed and $M>0$ is large enough.
We assume that $\lim_{n\to \infty} m/n=a$ and $0 \le a \le 1$.

Since 
$\frac{1}{n} \sum_{i=1}^{n} I_{\{Z_i>M\}}
 \to  P(V>M)$ and $\lim_{M\to \infty} P(V>M)=0$ and 
$$
\begin{aligned}
    \lim_{M\to \infty}  I_{\{Z_i<M\}} \zeta_i S_{\varepsilon_C |\varepsilon_T }(F_{\varepsilon_C }(H(Z_i)-W_i^T\eta)|F_{\varepsilon_T }(H(Z_i)-X_i^T\beta))=1, \\
    \lim_{M\to \infty}  I_{\{Z_i<M\}} \zeta_i S_{\varepsilon_T | \varepsilon_C }( F_{\varepsilon_T }(H(Z_i)-X_i^T\beta) | F_{\varepsilon_C }(H(Z_i)-W_i^T\eta))=1,
\end{aligned}
$$
we have that
\begin{equation}
\label{consistency:eta11}
    \limsup_{M\to \infty}  \limsup_{n\to \infty} \kappa_{1}(\hat{\theta}_n,\hat{H}_n)\le 0,
\end{equation}
by Lemma A.1 of \cite{Tong2019}.
Similarly,
\begin{equation}
\label{consistency:eta31}
    \limsup_{M\to \infty}  \limsup_{n\to \infty} \kappa_{3}(\hat{\theta}_n,\hat{H}_n)\le 0.
\end{equation}

Now we consider $t_k<Z_i\le t_{k+1}$ for $k=0,1,...,N-1$. Let $n_k$ denote the number of $\{i:Z_i\in(t_k, t_{k+1}]\}$.
Then when $n\to \infty$,
\begin{equation} \label{consistency:eta21}
    \begin{aligned}
&\quad \frac{1}{m} \sum_{i=1}^{m}I_{\{t_k<Z_i\le t_{k+1}\}}  l_{i}(Z_i;\hat{\beta}_n,\hat{H}_n) \\ 
&=\frac{1}{m} \sum_{i=1}^{m} I_{\{t_k<Z_i\le t_{k+1}\}} [\log  \omega (\hat{H}_n(Z_i)-X_i^T\hat{\beta}_n,
 \hat{H}_n(Z_i)-W_i^T\hat{\eta}_n)
+
\log\hat{H}_n\{Z_i\}+\log n] \\ 
&\le \frac{1}{m} \sum_{i=1}^{m} I_{\{t_k<Z_i\le t_{k+1}\}}  [\log \omega (\hat{H}_n(Z_i)-X_i^T\hat{\beta}_n,
\hat{H}_n(Z_i)-W_i^T\hat{\eta}_n)
+
\log \frac{\hat{H}_n(t_{k+1})-\hat{H}_n(t_k)}{n_k}  +\log n] \\ 
&=  E((\log \omega (H^*(V)-X^T\beta^*,
H^*(V)-W^T\eta^*)
+\log (H^*(t_{k+1})-H^*(t_k))) I_{\{t_k< V \le t_{k+1}\}} )
+p_k\log p_k +o(1),
    \end{aligned}
\end{equation}
where $p_k=P(t_k<V\le t_{k+1})$, $\lim_{n\to \infty} (\hat{\theta}_n,\hat{H}_n)= (\theta^*,H^*) $.

Let $\tilde{H}_n (t_k)=H_0(t_k) $ for $k=0,1,...,N$, and
$$
\begin{aligned}
    &\Delta \tilde{H}_n(Y_i)=\frac{\hat{H}_n(t_{k+1})-\hat{H}_n(t_{k})}{n_k}
\qquad
\text{for} \, V_i \in (t_k,t_{k+1}), k=0,1,...,N-1,\\
&\Delta \tilde{H}_n(Y_i)=\frac{1}{n}
\qquad
\text{for} \,V_i<-M \, \text{or} \, V_i>M.
\end{aligned}
$$
Considering $(\beta_0,H_0)$, similarly to \eqref{consistency:eta11} and \eqref{consistency:eta31}, we have,
\begin{equation}\label{consistency:eta12}
    \limsup_{M\to \infty}  \limsup_{n\to \infty} \kappa_{1}(\beta_0,\tilde{H}_n)\le 0 
\end{equation}
and
\begin{equation}\label{consistency:eta32}
    \limsup_{M\to \infty}  \limsup_{n\to \infty} \kappa_{3}(\beta_0,\tilde{H}_n)\le 0 .
\end{equation}
For $k=0,1,...,N-1$,
\begin{equation} \label{consistency:eta22}
    \begin{aligned}
& \frac{1}{m} \sum_{i=1}^{m}I_{\{t_k<Z_i\le t_{k+1}\}}  L_{i1}(Z_i;\beta_0,\tilde{H}_n) \\
&=\frac{1}{m} \sum_{i=1}^{m} I_{\{t_k<Z_i\le t_{k+1}\}}  [\log \omega(\tilde{H}_n(Z_i)-X_i^T\beta_0,
\tilde{H}_n(Z_i)-W_i^T\hat{\eta}_n)
+
\log H_0\{Z_i\}+\log n] \\  
&=  E((\log \omega (H_0(V)-X^T\beta_0, H_0(V)-W^T \eta_0) )
I_{\{t_k< V \le t_{k+1}\}} )
+p_k\log (H_0(t_{k+1})-H_0(t_k)) +p_k\log p_k +r_{n,k},
    \end{aligned}
\end{equation}
where $r_{n,k}$ is such that
$\limsup_{n\to \infty} \sum_{k=0}^{N-1} |r_{n,k}|=0$,
as $\lim_{N\to \infty} \max_{-\infty\le k \le N-1} |t_{k+1}-t_k|\to 0$.

Then by (\ref{consistency:eta21}) and (\ref{consistency:eta22}), we have
$$
 \begin{aligned}
& \limsup_{n\to \infty}  (\kappa_{2}(\hat{\beta}_n,\hat{H}_n)-\kappa_{2}(\beta_0, \tilde{H}_n)) \\
&\le\sum_{k=0}^{N-1}   E(I_{\{t_k< V \le t_{k+1}\}} \times (\log \omega(H^*(V)-X^T\beta^*,H^*(V)-W^T\eta^*)
+\log (H^*(t_{k+1})-H^*(t_k)))) \\
&\quad - \sum_{k=0}^{N-1}  E(I_{\{t_k<V\le t_{k+1}\}} \times 
(\log \omega(H_0(V)-X^T\beta_0,H_0(V)-W^T\eta_0) 
+\log (H_0(t_{k+1})-H_0(t_k)))) +\limsup_{n\to \infty}\sum_{k=0}^{N-1}|r_{n,k}|  \\
&=J_1(\theta^*,H^*,\theta_0,H_0)
+\limsup_{n\to \infty}\sum_{k=0}^{N-1}|r_{n,k}|, \\
 \end{aligned}
$$
where $\alpha_N=P(t_0<V \le t_N)$.

Then we have 
$$
J_1(\theta^*,H^*,\theta_0,H_0)=
E( I_{\{t_0<V\le t_{N}\}}          
\log \sum_{k=0}^{N-1}I_{\{t_k<V\le t_{k+1}\}}   \frac{\omega(H^*(V)-X^T\beta^*,H^*(V)-W^T\eta^*)
(H^*(t_{k+1})-H^*(t_k))}{\omega(H_0(V)-X^T\beta_0,H_0(V)-W^T\eta_0)
(H_0(t_{k+1})-H_0(t_k))}),
$$
and
$$
J_1(\theta^*,H^*,\theta_0,H_0)
\le 
\alpha_N\log E \sum_{k=0}^{N-1}I_{\{t_k<V\le t_{k+1}\}} \frac{\omega (H^*(V)-X^T\beta^*,H^*(V)-W^T\eta^*)
(H^*(t_{k+1})-H^*(t_k))}{\omega (H_0(V)-X^T\beta_0,H_0(V)-W^T\eta_0)
(H_0(t_{k+1})-H_0(t_k))}
- \alpha_N \log \alpha_N.
$$

Let $-M=t_0<t_1<...<t_N=M$ such that
$\lim_{N\to \infty} \max_{-\infty\le k \le N-1} |t_{k+1}-t_k|\to 0 $.
Then using Lemma A.3 of \cite{Tong2019}, 
\begin{equation} \label{consistency:alpha}
\begin{aligned}
&\log E \sum_{k=0}^{N-1}I_{\{t_k<V\le t_{k+1}\}} 
\frac{\omega(H^*(V)-X^T\beta^*,H^*(V)-W^T\eta^*)(H^*(t_{k+1})-H^*(t_k))}
{\omega(H_0(V)-X^T\beta_0, H_0(V)-W^T\eta_0)(H_0(t_{k+1})-H_0(t_k))}\\
&=\log \int_{\mathcal{X} } \int_{\mathcal{W} } \sum_{k=0}^{N-1}
\zeta_k
\int_{t_k}^{t_{k+1}} \frac{\omega(H^*(t)-x^T\beta^*,H^*(t)-w^T\eta^*)(H^*(t_{k+1})-H^*(t_k))}
{\omega(H_0(t)-x^T\beta_0,H_0(t)-w^T\eta_0)(H_0(t_{k+1})-H_0(t_k))} \\
& \qquad\times \omega (H_0(t)-x^T\beta_0,H_0(t)-w^T\eta_0)h_0(t)f_{X,W}(x,w)dtdwdx\\
&\to \log \int_{\mathcal{X} } \int_{\mathcal{W} } \int_{-M}^{M} a \omega (H^*(t)-x^T\beta^*,H^*(t)-w^T\eta^*) f_{X,W}(x,w)
dH^*(t) dw dx
\; as \; N\to \infty ,
 \end{aligned}
\end{equation} 
which converges to 0 as $M\to \infty$.
$\mathcal{X}$ and $\mathcal{W}$ represent the compact sets of the variables $X$ and $W$ respectively.
Thus we have
$\limsup_{M\to \infty} \limsup_{N\to \infty} \limsup_{n\to \infty}
(l_n(\hat{\theta}_n,\hat{H}_n)^{(1)}-l_n(\theta_0,\tilde{H}_n)^{(1)})
=0$.

Besides, for $l_n(\theta, H)^{(2)}$,  by L'Hopital rules of limits, we have
$$
\lim_{t \to \infty}
\frac{ \tilde{C} \left( H^*(t) - X_i^\top \beta^* ,\, H^*(t) - W_i^\top \eta^* \right)}
{\tilde{C} \left( H_0(t) - X_i^\top \beta_0 ,\, H_0(t) - W_i^\top \eta_0 \right)}
= \lim_{t \to \infty} \frac{ \omega \left( H^*(t) - X_i^\top \beta^* ,\, H^*(t) - W_i^\top \eta^* \right)}
{\omega \left( H_0(t) - X_i^\top \beta_0 ,\, H_0(t) - W_i^\top \eta_0 \right)}, $$
Since $ 0 \le \tilde{C} \left( H(t) - X_i^\top \beta ,\, H(t) - W_i^\top \eta \right) \le  1$, we obtain that
$\limsup_{M\to \infty} \limsup_{N\to \infty} \limsup_{n\to \infty}
(l_n(\hat{\theta}_n,\hat{H}_n)^{(2)}-l_n(\theta_0,\tilde{H}_n)^{(2)})
=0$.
Therefore, we have that
\begin{equation}
\limsup_{M\to \infty} \limsup_{N\to \infty} \limsup_{n\to \infty}
(l_n(\hat{\theta}_n,\hat{H}_n)-l_n(\theta_0,\tilde{H}_n))
=0.
\label{consistency:like}
\end{equation}

(Step 2)
We aim to prove $ H_0$ is absolutely continuous with respect to $H^*$ by considering the following equation

\begin{equation}\label{consistency:reconsi}
    \begin{aligned}
  J_2(\theta^*,H^*,\theta_0,H_0) 
  & = \int_{\mathcal{X} } \int_{\mathcal{W} } \sum_{k=0}^{N-1} \zeta_k \int_{t_k}^{t_{k+1}} \log \frac{\omega_(H^*(t)-x^T\beta^*,H^*(t)-w^T\eta^*)(H^*(t_{k+1}) - H^*(t_k))}
 {\omega_(H_0(t)-x^T\beta_0,H_0(t)-w^T\eta_0)(H_0(t_{k+1}) - H_0(t_k))} \\
 & \qquad \times \omega_(H_0(t)-x^T\beta_0,H_0(t)-w^T\eta_0)
 f_{X,W}(x,w)
 \; dt dw dx. 
\end{aligned}
\end{equation}
It is easy to prove that
$$
 J_1(\theta^*,H^*,\theta_0,H_0)= J_2(\theta^*,H^*,\theta_0,H_0).
$$
By Step 1, we obtain 
$$
\begin{aligned}
&\limsup_{M\to \infty} \limsup_{N\to \infty}|  J_2(\theta^*,H^*,\theta_0,H_0)|=0.
 \end{aligned}
 $$
Given that $\omega$ and $h_0$ each have a positive lower bound and a finite upper bound on any finite interval, the above equation shows that for any
 $M\in (0,\infty)$, we have
 $$
 \sup_{N\ge 1} | \sum_{k=0}^{N-1} (H_0(t_{k+1})-H_0(t_k))
\log \frac{(H_0(t_{k+1})-H_0(t_k))}{(H^*(t_{k+1})-H^*(t_k))}| <\infty . $$
Then, using  $xlogx=\sup_{\alpha\in \mathcal{R}}\{ \alpha x - e ^{\alpha-1} \}$, we obtain that for any $r>0$,
$$
\begin{aligned}
&| \sum_{k=0}^{N-1} (H_0(t_{k+1})-H_0(t_k))
\log \frac{(H_0(t_{k+1})-H_0(t_k))}{(H^*(t_{k+1})-H^*(t_k))}| \\
&= | \sum_{k=0}^{N-1} \sup_{\alpha\in \mathcal{R}}\{\alpha (H_0(t_{k+1})-H_0(t_k)) -  (H^*(t_{k+1})-H^*(t_k)) e ^{\alpha-1} \} | \\   
&\ge r \sum_{k=0}^{N-1} |H_0(t_{k+1})-H_0(t_k) | - e ^{r-1} \sum_{k=0}^{N-1}|H^*(t_{k+1})-H^*(t_k) |.
\end{aligned}
$$
If 
$\lim_{N\to \infty} \sum_{k=0}^{N-1}|H^*(t_{k+1})-H^*(t_k) |=0$,
we obtain that
$$
\begin{aligned}
\limsup_{N\to \infty}  \sum_{k=0}^{N-1} |H_0(t_{k+1})-H_0(t_k) |
&\le \frac{1}{r} \sup_{N\ge 1} | \sum_{k=0}^{N-1} (H_0(t_{k+1})-H_0(t_k))
\log \frac{(H_0(t_{k+1})-H_0(t_k))}{(H^*(t_{k+1})-H^*(t_k))}| \\
&\to 0 , \; as \; r\to \infty.
\end{aligned}
$$
Thus, $H_0$ is absolutely continuous with respect to $H^*$. 
For any $M>0$ , applying (\ref{consistency:reconsi}) , we can derive that
\begin{equation}
\begin{aligned}
\limsup_{N\to \infty}  J_2(\theta^*,H^*,\theta_0,H_0) 
&= \int_{\mathcal{X} } \int_{\mathcal{W} } \int_{-M}^{M} a \log \frac{f^*_{V|x,w}(t)}{f^0_{V|x,w}(t)}f^0_{V|x,w}(t) f_{X,W}(x,w) dt dw dx \\
&=- a E( \int_{-M}^{M}  \log \frac{f^0_{V|X,W}(t)}{f^*_{V|X,W}(t)} f^0_{V|X,W}(t) dt ),
\end{aligned}
\label{proof:star3}
\end{equation}
where 
$$
\begin{aligned}
f^*_{V|x,w}(t)=\omega(H^*(t)-x^T \beta^*, H^*(t)-w^T\eta^*) h^*(t),\\
f^0_{V|x,w}(t)=\omega(H_0(t)-x^T \beta_0, H_0(t)-w^T\eta_0) h_0(t),
\end{aligned}
$$
and $h^*(t)=dH^*(t)/dt$, $h_0(t)=dH_0 (t)/dt$.

(Step3)
We define $ F^*_{V|X,W}(t)=\int_{-\infty}^{t} f^*_{V|X,W}(v) dv   $, $ F^0_{V|X,W}(t)=\int_{-\infty}^{t} f^0_{V|X,W}(v) dv $.
By Theorem \ref{theorem:identification},  to obtain the consistency of $(\hat{\theta}_n,\hat{H}_n)$, it is sufficient for us to prove $(\theta^*,H^*)=(\theta_0,H_0)$  by proving $f^*_{V|X,W}(t)=f^0_{V|X,W}(t)$.

We start from noting that $|\log x |\le 1/x$ for all $0<x\le 1$. We have that
$(\log \frac{f^0_{V|X,W}(t)}{f^*_{V|X,W}(t)})^- \le ( \frac{f^0_{V|X,W}(t)}{f^*_{V|X,W}(t)})^{-1}$.
Thus 
$$
\int_{-M}^{M} (\log \frac{f^0_{V|X,W}(t)}{f^*_{V|X,W}(t)})^{-} f^0_{V|X,W}(t) dt 
\le F^*_{V|X,W}(t) \le 1.
$$
According to the monotone convergence theorem, we have
$$
\lim_{M\to \infty }
E(\int_{-M}^{M} (\log \frac{f^0_{V|X,W}(t)}{f^*_{V|X,W}(t)})^+ f^0_{V|X,W}(t) dt)
=E(\int_{-\infty}^{\infty} (\log \frac{f^0_{V|X,W}(t)}{f^*_{V|X,W}(t)})^+ f^0_{V|X,W}(t) dt ),
$$
and
$$
\lim_{M\to \infty }
E(\int_{-M}^{M} (\log \frac{f^0_{V|x,w}(t)}{f^*_{V|x,w}(t)})^- f^0_{V|X,W}(t) dt)
=E(\int_{-\infty}^{\infty} (\log \frac{f^0_{V|X,W}(t)}{f^*_{V|X,W}(t)})^- f^0_{V|X,W}(t) dt )
<\infty.
$$
Thus, 
$$
\begin{aligned}
-E(\int_{-\infty}^{\infty} (\log \frac{f^0_{V|X,W}(t)}{f^*_{V|X,W}(t)}) f^0_{V|X,W}(t) dt )
=E(\int_{-\infty}^{\infty} (\log \frac{f^*_{V|X,W}(t)}{f^0_{V|X,W}(t)}) f^0_{V|X,W}(t) dt )
\le 0.
\end{aligned}
$$
By (\ref{consistency:like}),  (\ref{consistency:reconsi}) and (\ref{proof:star3}), we obtain
$$
\begin{aligned}
 -\lim_{M\to \infty } a E(\int_{-M}^{M} (\log \frac{f^0_{V|X,W}(t)}{f^*_{V|X,W}(t)}) f^0_{V|X,W}(t) dt )
 =\limsup_{M \to \infty}\limsup_{N \to \infty} J_2(\theta^*,H^*,\theta_0,H_0)
 \ge 0.
\end{aligned}
$$
The above equation implies that
$$
\begin{aligned}
0 &\le -\lim_{M\to \infty } E(\int_{-M}^{M} (\log \frac{f^0_{V|X,W}(t)}{f^*_{V|X,W}(t)}) f^0_{V|X,W}(t) dt )
=- E(\int_{-\infty}^{\infty} (\log \frac{f^0_{V|X,W}(t)}{f^*_{V|X,W}(t)}) f^0_{V|X,W}(t) dt )
\le 0, 
\end{aligned}
$$
which demonstrate that
$f^*_{V|X,W}(t)=f^0_{V|X,W}(t)$.
Hence, for any $t \in \mathcal{R}$,
$$
 \omega(H^*(t)-X^T \beta^*, H^*(t)-W^T\eta^*) h^*(t)
= \omega(H_0(t)-X^T \beta_0, H_0(t)-W^T\eta_0) h_0(t).
$$
Since we have  established  the identifiability of the model,  it follows that
$\theta^*=\theta_0$, $H^*=H_0$. Thus we complete the proof.

$\hfill\blacksquare$

%% file: sections/AppendixC_normality.tex
\section{Proof of Theorem \ref{theorem:normal} }\label{sec:mynorm2}

To establish the asymptotic normality of the estimators, we first define $\{T_n\}$ such that $P(|Z| \ge T_n) = o(n^{-3/4})$ and $E\left[\dot{\psi}(H_0(Z)-X^T \beta_0, H_0(Z)-W^T \eta_0)H_0(Z)I(|Z| \ge T_n)\right] = o_p(n^{-3/4})$, where $\dot{\psi}(x_1,x_2)$ is given by $\partial \psi(x_1,x_2) /\partial x_1 + \partial \psi(x_1,x_2) /\partial x_2$.
Additionally, for any $T > 0$ and $s=(\theta, H)$, the pseudo distance between  $s_1$ and $s_2$ is defined as
$$
d_T(s_1,s_2)=\sup_{|t|\le T }
|H_1(t)-H_2(t)| + ||\theta_1 - \theta_2 ||.
$$
From Lemma \ref{lemma:distance}, we can obtain that $d_{T_n}(\hat{s}_n,s_0)=o_p(1)$.
For clarity of proof,  we use the function \(\psi (\tilde{H}_n(z)-X^T\hat{\beta}_n,\tilde{H}_n(z)-W^T\hat{\eta}_n)\) as an abbreviation of \eqref{eq:alg_psi}.

Then we follow the  steps similar to the proof of Theorem 2.3 of \cite{Tong2019}.
We start by  defining  the truncated function of  $\hat{H}_n$ at $T_n$ is that
 $$
\tilde{H}_n(t)= \begin{cases}   
\hat{H}_n(-T_n),\quad &t\le -T_n, \\
\hat{H}_n(t),\quad &|t|\leq T_n ,\\
\hat{H}_n(T_n),\quad &t\ge T_n. \\
\end{cases}
 $$

For simplicity of notation, we will denote $\tilde{\theta}_n = \hat{\theta}_n$ and $\tilde{s}_n = (\tilde{\theta}_n, \tilde{H}_n)$ in the following proof.
Given that the score function for $\theta$  at $\tilde{s}_n$ is zero, we have 
$P_n U(Z,\Delta,\xi,\hat{\theta}_n,\hat{H}_n)=0$.
Consequently, by   the proof of Lemma A.4 of \cite{Tong2019},  we have
$$
\begin{aligned}
o_p(n^{-1/2})
&=P_n I(|Z|\le T_n) U(Z,X,W,\Delta,\xi,\hat{\theta}_n,\hat{H}_n) \\
&=P_n U(Z,X,W,\Delta,\xi,\theta_0,H_0) + 
P_n \{I(|Y|\le T_n)
\frac{\partial U(Z,X,W,\Delta,\xi,\theta_0,H_0)}{\partial H}
[\hat{H}_n(Y)-H_0(Y)]    \\
&\quad +P_n \{I(|Y|\le T_n) \frac{\partial U(Z,X,W,\Delta,\xi,\theta_0,H_0)}{\partial \theta}(\hat{\theta}_n-\theta_0)
+ o_p(d_{T_n}(\hat{s}_n,s_0)).
\end{aligned}
$$
By substituting $\hat{H}_n$ with $\tilde{H}_n$, the above equation remains valid.
Since 
$$
\begin{aligned}
    U(Z,\Delta,\xi,\theta,H)
&=\sum_{i=1}^{n}\partial L_i(\theta,H)/\partial \theta\\
&=\sum_{i=1}^{n} \Delta_i(1-\xi_i)\frac{\partial  L_{i1}(\theta,H)}{\partial \theta} 
+ \sum_{i=1}^{n} (1-\Delta_i)\xi_i  \frac{\partial L_{i2}(\theta,H)}{\partial \theta} 
+ \sum_{i=1}^{n} (1-\Delta_i)(1-\xi_i) \frac{\partial L_{i3}(\theta,H)}{\partial \theta} \\
&=:U_1 + U_2 + U_3 
\end{aligned}
$$
and in addition,
$$
\begin{aligned}
&U'_{\theta}(Z,\Delta,\xi,\theta_0,H_0)
=\frac{\partial U(Z,X,W,\Delta,\xi,\theta_0,H_0)}{\partial \theta} , \\
&  U'_{H}(Z,\Delta,\xi,\theta_0,H_0)
=\frac{\partial U(Z,X,W,\Delta,\xi,\theta_0,H_0)}{\partial H}
=:U'_{1H}+U'_{2H}+U'_{3H}. 
\end{aligned}
$$
We obtain that
\begin{equation}\label{norm:A12}
    E  U'_{\theta}(Z,\Delta,\xi,\theta_0,H_0) (\hat{\theta}_n-\theta_0) + \mathcal{A} [\tilde{H}_n-H_0]
= - P_n  U(Z,\Delta,\xi,\theta_0,H_0) 
+o_p(d_{T_n}(\hat{s}_n,s_0)) +o_p(n^{-1/2}),
\end{equation}
where $\mathcal{A}$ is the functional operator such that 
$$
\begin{aligned}
\mathcal{A}[H]
&=
\int_{-\infty}^{\infty} E_{X,W}[U'_{1H}f_{Z,\Delta,\xi|X,W}(t,1,0) ]H(t) \mathrm{d}H_0(t) +
\int_{-\infty}^{\infty} E_{X,W}[U'_{2H}f_{Z,\Delta,\xi|X,W}(t,0,1) ]H(t) \mathrm{d}H_0(t) \\
&\quad + \int_{-\infty}^{\infty} E_{X,W}[U'_{3H}f_{Z,\Delta,\xi|X,W}(t,0,0) ]H(t) \mathrm{d}H_0(t) \\
&=:\mathcal{A}_1 [H] + \mathcal{A}_2 [H] + \mathcal{A}_3 [H], 
\end{aligned}
$$
and $H\in \{ H:|H(t)|\le 3 H_0(t) \}$.
Next, we analyze the asymptotic representation of
$\tilde{H}_n(t)$. 
To begin with, 
\begin{equation}
    \tilde{H}_n(t)=\int_{0}^{t} 
    \frac{
    P_n dN_+(z) }
    {P_n  Y(z) \psi (\tilde{H}_n(Z)-X^T\hat{\beta}_n,\tilde{H}_n(Z)-W^T\hat{\eta}_n) }, 
    \label{eq:asym_repre}
\end{equation}
where $N_+(z)=\delta I(0<Z\le z)$, $\delta=I(V\le A)$ , 
$Y(z)=I(Z>z)$.
We define the martingale process $M_+(t)$ as
$$
dM_+(z)=dN_+^*(z)-
Y_i(z) \lambda( H(z)-X^T\beta ,  H(z)-W^T\eta).
$$
where
$$\lambda(z)
=\frac{S^{(1)}( H(z)-X^T\beta ,  H(z)-W^T\eta)
+S^{(2)}( H(z)-X^T\beta,  H(z)-W^T\eta)}
{S( H(z)-X^T\beta ,  H(z)-W^T\eta) },
$$
with 
$S(z_1,z_2)=\tilde{\mathcal{C}}(F_{\varepsilon _T}(z_1),F_{\varepsilon _C}(z_2))$, 
$S^{(1)}(z_1,z_2)=
dS(z_1,z_2) / d z_1$
and 
$S^{(2)}(z_1,z_2)=
d S(z_1,z_2) / d z_1$.

To simplify, we write 
\begin{equation}
\begin{aligned}
    & \hat{\varepsilon}=(\hat{\varepsilon} _T,\hat{\varepsilon} _C)
    = (\hat{H}_n(Z)-X^T\hat{\beta}_n, \hat{H}_n(Z)-W^T\hat{\eta}_n), \\
    & \varepsilon(t)=(\varepsilon_T(t), \varepsilon_C(t))
    =(H_0(t)-X^T \beta_0, H_0(t)-W^T \eta_0).
\end{aligned}
\label{eq:note_abb}
\end{equation}
Therefore, based on the proof of  Lemma A.4 of \cite{Tong2019} and lemma \ref{lemma:distance},  we find that  for any $z\in[0,T_n]$,
\begin{equation}
    \begin{aligned}
\tilde{H}_n(t)-H_0(t)
&=\int_{0}^{t} \frac{P_n dN_+(z)}
{P_n \{ Y(z) \psi (\tilde{H}_n(Z)-X^T\hat{\beta}_n,\tilde{H}_n(Z)-W^T\hat{\eta}_n)\}} - H_0(t)\\
&=\int_{0}^{t} \frac{P_n dM_+(z)}{P_n \{ Y(z) \psi (\tilde{H}_n(Z)-X^T\hat{\beta}_n,\tilde{H}_n(Z)-W^T\hat{\eta}_n)\} } 
+ \int_{0}^{t} \frac{P_n I(Y\ge z)[\lambda(\varepsilon (z))
- \psi(\varepsilon )]dH_0(z)}
{P_n \{ Y(s) \psi (\tilde{H}_n(Z)-X^T\hat{\beta}_n,\tilde{H}_n(Z)-W^T\hat{\eta}_n)\} }   \\
&\quad - \int_{0}^{t} \frac{P_n  Y(z) [ \psi (\tilde{H}_n(Z)-X^T\hat{\beta}_n,\tilde{H}_n(Z)-W^T\hat{\eta}_n)
-\psi(H_0(Z)-X^T \beta_0,H_0(Z)-W^T \eta_0)] dH_0(z)}
{P_n \{ Y(z) \psi(\tilde{H}_n(Z)-X^T\hat{\beta}_n,\tilde{H}_n(Z)-W^T\hat{\eta}_n)\}  }.   \\
\end{aligned}
\label{eq:HnH0minus}
\end{equation}
The last term of equation \eqref{eq:HnH0minus} satisfies
$$
\begin{aligned}
&\int_{0}^{t} \frac{P_n  Y(z) [ \psi(\tilde{H}_n(Z)-X^T\hat{\beta}_n,\tilde{H}_n(Z)-W^T\hat{\eta}_n)
- \psi(H_0(Z)-X^T \beta_0,H_0(Z)-W^T \eta_0)] dH_0(z)}
{P_n  Y(z) \psi (\tilde{H}_n(Z)-X^T\hat{\beta}_n,\tilde{H}_n(Z)-W^T\hat{\eta}_n) } \\ 
&=\int_{0}^{t} \frac{P_n Y(z) 
[\dot{\psi} (\varepsilon )(\hat{H}_n(Z)-H_0(Z))
+\psi'_{\omega} \omega'_{\theta}(\hat{\theta}_n-\theta_0)
]}
{E Y(s) \psi(\varepsilon_T,\varepsilon_C)}dH_0(s) + o_p(d_{T_n}(\tilde{s}_n, s_0) ) + o_p(n^{-1/2}) , \\ 
\end{aligned}
$$
where $\dot{\psi} (x_1,x_2)=\partial \psi(x_1,x_2) /\partial x_1 + \partial \psi(x_1,x_2) /\partial x_2 $, 
$\psi'_{\omega}=\partial \psi(\varepsilon_T,\varepsilon_C) / \partial \omega  $,
$\omega'_{\theta}=\partial \omega(\varepsilon_T,\varepsilon_C)/ \partial \theta  $.
Besides, 
define 
$$
\begin{aligned}
& \mathcal{J}_1^+(t)=\int_{0}^{t} \frac{ dM_+(z)}{E Y(z) \psi(\varepsilon) } 
+ \int_{0}^{t} \frac{ Y(z)[\lambda(\varepsilon (z)-\psi(\varepsilon)  ]]dH_0(z)}
{E Y(z) \psi(\varepsilon) } \\
&\mathcal{J}_2^+(t)=
\int_{0}^{t} \frac{E Y(z) \psi'_{\omega} \omega'_{\theta}}
{E Y(z) \psi(\varepsilon)} dH_0(z),
\end{aligned}
$$
and the functional operator
$$
\begin{aligned}
\mathcal{G}_+[H](t)
&=\int_{0}^{\infty} H(z)E_{X,W}[\dot{\psi} (\varepsilon (z))f(\varepsilon (z)) ]dH_0(z)
\int_{0}^{z} \frac{I(u\le t)}{EI(Y\ge z) \psi(\varepsilon )} du. \\
\end{aligned}
$$
Then we obtain that
$$
\tilde{H}_n(t)-H_0(t)=P_n \mathcal{J}_1^+(t)-\mathcal{J}_2^+(t)'(\hat{\theta}_n-\theta_0)-\mathcal{G}_
+[\tilde{H}_n-H_0](t)+o_p(d_{T_n}(\tilde{s}_n,s_0)+n^{-1/2}).
$$
Similarly, define $U_1^-(t)$, $U_2^-(t)$ and  $\mathcal{G}_-$. Then we obtain that for $t\in [-T_n,0]$,
$$
\tilde{H}_n(t)-H_0(t)=P_n \mathcal{J}_1^-(t) - \mathcal{J}_2^-(t)'(\hat{\theta}_n-\theta_0)-\mathcal{G}_
-[\tilde{H}_n-H_0](t)+o_p(d_{T_n}(\tilde{s}_n,s_0)+n^{-1/2}).
$$
Therefore, it follows that for $t\in [-T_n,T_n]$,
$$
\tilde{H}_n(t)-H_0(t)=P_n \mathcal{J}_1(t) - \mathcal{J}_2(t)'(\hat{\theta}_n-\theta_0)-\mathcal{G}[\tilde{H}_n-H_0](t)+o_p(d_{T_n}(\tilde{s}_n,s_0)+n^{-1/2}),
$$
where $\mathcal{J}_1(t)=\mathcal{J}_1^-(t)+ \mathcal{J}_1^+(t)$ and so are $U_2(t)$ and $\mathcal{G}$.
Then we have
$$
\tilde{H}_n(\cdot)-H_0(\cdot)=
P_n(I+\mathcal{G})^{-1}[\mathcal{J}_1](\cdot)-(I+\mathcal{G})^{-1}[\mathcal{J}_2](\cdot)'(\hat{\theta}_n-\theta_0)
+o_p(d_{T_n}(\tilde{s}_n,s_0)+n^{-1/2}).
$$
Thus with (\ref{norm:A12}), we obtain that 
$$
\begin{aligned}
&[E\{ U'_{\theta}  \} - \mathcal{A} [(I+\mathcal{G})^{-1}[\mathcal{J}_2](Z)]'](\hat{\theta}_n-\theta_0)\\ 
&=-P_n [ U +[(I+\mathcal{G})^{-1}[\mathcal{J}_1](Z)]
+o_p(d_{T_n}(\hat{s}_n,s_0)) +o_p(n^{-1/2})
\end{aligned}
$$
and thus that
$$
\begin{aligned}
\sqrt{n}(\hat{\theta}_n-\theta_0)
&=-[E\{ U'_{\theta}  \} - \mathcal{A} [(I+\mathcal{G})^{-1}[\mathcal{J}_2](Z)]']^{-1} \sqrt{n} P_n \times\\
&\quad  [U +[(I+\mathcal{G})^{-1}[\mathcal{J}_1](Z)]
+o_p(\sqrt{n}d_{T_n}(\hat{\theta}_n,\theta_0)) +o_p(1)
\end{aligned}
$$
Therefore, similar to Theorem 2.3 of \cite{Tong2019},  $|| \tilde{H}_n - H_0 ||=O_p(|\tilde{\theta}_n -\theta_0|+n^{-1/2} )  $ implies that
$$
\sqrt{n}(\hat{\theta}_n-\theta_0) \to N(0,\Sigma),
$$
where 
$$
\Sigma=A_{\theta}^{-1} \Sigma_{\theta} A_{\theta}^{{-1}^T},
$$
with
\begin{equation}
    A_\theta=E\{ U'_{\theta}  \} - \mathcal{A} [(I+\mathcal{G})^{-1}[\mathcal{J}_2](Z)]',
\label{Atheta}
\end{equation}
\begin{equation}
    \Sigma_\theta=cov[ U +[(I+\mathcal{G})^{-1}[\mathcal{J}_1](Z) ].
\label{Sigmatheta}
\end{equation}
Then we complete the proof.

$\hfill\blacksquare$ 

%% file: sections/AppendixD_Lemma.tex
\section{Lemma }\label{sec:Lemma}


\begin{lemma}
Assume that Conditions (C\ref{C1})-(C\ref{C7}) given above hold and that Asssumptions (A\ref{A1})-(A\ref{A8}) are satisfied. 
Suppose that
$$\lim_{z \to a}
\frac{H_1(z)}{H_2(z)}=\frac{1}{c}
,\quad
\lim_{z \to a}
\frac{g_{\varepsilon _T,\lambda_{T,1}}(H_1(z)-x^T\beta_1)}{g_{\varepsilon _T,\lambda_{T,2}}(H_2(z)-x^T\beta_2)}=c, \quad 
\lim_{z \to a}
\frac{f_{\varepsilon _T,\lambda_{T,1}}(H_1(z)-x^T\beta_1)}{f_{\varepsilon _T,\lambda_{T,2}}(H_2(z)-x^T\beta_2)}=c, 
$$
where $0<c<\infty $, $a=-\infty$ or $\infty$.
Then if $c=1$, we have
$\lim_{z \to a} \{ H_1(z) - H_2(z) \}=0$, $\beta_1=\beta_2$.
\label{lemma_H1H2}
\end{lemma}

\textbf{Proof}:
Without loss of generity, we assume that $a=-\infty$.
Since $c=1$, we have
$$
\lim_{z \to -\infty} [\log f_{\varepsilon _T,\lambda_{T,1}}(H_1(z)-x^T\beta_1) - 
\log f_{\varepsilon _T,\lambda_{T,2}}(H_2(z)-x^T\beta_2)]=0.
$$
It now follows that
\begin{equation}
\begin{aligned}
  0 
&=\lim_{z \to - \infty} [\log f_{\varepsilon _T,\lambda_{T,1}}(H_1(z)-x^T\beta_1) - 
\log f_{\varepsilon _T,\lambda_{T,2}}(H_2(z)-x^T\beta_2)] \\
&=\lim_{z \to - \infty} [\log f_{\varepsilon _T,\lambda_{T,1}}(H_1(z)-x^T\beta_1) - 
\log f_{\varepsilon _T,\lambda_{T,1}}(H_2(z)-x^T\beta_2)] \\
 & \qquad \qquad  + \lim_{z \to - \infty} [\log f_{\varepsilon _T,\lambda_{T,1}}(H_2(z)-x^T\beta_2) - 
\log f_{\varepsilon _T,\lambda_{T,2}}(H_2(z)-x^T\beta_2)] \\
& =\lim_{z \to - \infty} [G_1(z,x) +  G_2(z,x)]. 
\end{aligned}
\label{lemma:part}
\end{equation}
Since the limit of $g_{\varepsilon_T, \lambda_{T,j}}(z)$ is finite for $j = 1, 2$ as $z \to -\infty$, we have
$$
1 = \lim_{z \to -\infty} \frac{g_{\varepsilon_T, \lambda_{T,1}}(z)}{g_{\varepsilon_T, \lambda_{T,2}}(z)}
= \frac{\lim_{z \to -\infty} g_{\varepsilon_T, \lambda_{T,1}}(H_1(z) - x^T\beta_1)}
{\lim_{z \to -\infty} g_{\varepsilon_T, \lambda_{T,2}}(H_2(z) - x^T\beta_2)},
$$
which implies that $\lambda_{T,1} = \lambda_{T,2}$, or alternatively, $\lim_{t \to -\infty} \frac{f_{\varepsilon_T, \lambda_{T,1}}(t)}{f_{\varepsilon_T, \lambda_{T,2}}(t)} = 1$ by Assumption (A\ref{A7}).
In both cases, this  indicates that $\lim_{z \to -\infty} G_2(z, x) = 0$. Consequently, we have $\lim_{z \to -\infty} G_1(z, x) = 0$.
By Lagrange mean value theorem,
we have that
\begin{equation}
\begin{aligned}
 &\lim_{z \to - \infty} G_1(z,x)
&=\lim_{z \to - \infty}
[ g_{\varepsilon _T,\lambda_{T,1}}(V^* )(H_1(z) - H_2(z) +
x^T(\beta_2 - \beta_1))], 
\end{aligned}
\label{lemma:eq1}
\end{equation}
where $V^*$ is the value between $\{H_1(z)-x^T\beta_1 \}$ and $ \{ H_2(z)-x^T\beta_2 \}$ and $\lim_{z \to -\infty}g_{\varepsilon _T,\lambda_{T,1}}(z )=g_{-\infty}>0 $.
Since  $\lim_{z \to - \infty} G_1(z,x) =0$,
we obtain that
$$
\lim_{z \to - \infty}
[ H_1(z) - H_2(z) +
x^T(\beta_2 - \beta_1)] = 0. 
$$
Hence we have
$
 x^T(\beta_2 - \beta_1)
=\lim_{z \to - \infty} 
\{ H_1(z) - H_2(z) \}
$
for almost every $x$, which
means that $Var(x^T(\beta_1 - \beta_2)) = 0 $.
 It now follows from Assumption (A\ref{A4}) that $\beta_1 = \beta_2$, and hence we
also have that $\lim_{z\to -\infty}\{H_2(z)-H_1(z)\} = 0$.
Thus the proof is completed.

$\hfill\blacksquare$


\begin{lemma}
    Assume that Conditions (C\ref{C1})-(C\ref{C5}) given above hold true and that Assumptions (A\ref{A1})-(A\ref{A8}) are satisfied. 
    Then $d_{T_n}(\hat{s}_n,s_0)=o_p(1)$.
    \label{lemma:distance}
\end{lemma}

\textbf{Proof}:
Given the conditions and assumptions, 
from Theorem \ref{theorem:consistency}, we have
$|\hat{\theta}_n -\theta_0|=o_p(1) $.
Thus, it is sufficient for us to prove that $\sup_{|z|\le T_n} |\hat{H}_n(z) - H_0(z)|=o_p(1)$. Considering the asymptotic representation $\hat{H}_n(t)$ in \eqref{eq:asym_repre}, we have
\begin{equation}
    \hat{H}_n(t)=\int_{0}^{t} 
    \frac{P_n \, dI(Z \le z)}
    {P_n \, I(Z \ge z) \, \psi (\hat{H}_n(Z)-X^T\hat{\beta}_n, \, \hat{H}_n(Z)-W^T\hat{\eta}_n)}.
    \label{eq:lem:pres}
\end{equation}
By L'H\^{o}pital's rule, as $z \to \infty$, we have
$$
\begin{aligned}
    \lim_{t \to \infty}
    \frac{\partial }{\partial H(z)} \log S
    \left( H(z)-X^T\beta ,\, H(z)-W^T\eta \right)
    & = - \lim_{t \to \infty} \frac{\omega\left( H(z)-X^T\beta,\, H(z)-W^T\eta\right) }{S \left( H(z)-X^T\beta ,\, H(z)-W^T\eta \right) } \\
    & =\lim_{t \to \infty}
    \frac{\partial }{\partial H(z)} \log \omega
    \left( H(z)-X^T\beta ,\, H(z)-W^T\eta \right),
\end{aligned}
$$
while when $z \to -\infty$, we have
$$
\lim_{t \to \infty}
\frac{\partial }{\partial H(z)} \log S
\left( H(z)-X^T\beta ,\, H(z)-W^T\eta \right) =0.
$$
Since 
$$
\begin{aligned}
    \frac{d \varpi (z)}{dz} 
&= \frac{d }{dz}[ f_{T }(z)(1 - G_1(z)) + f_{C }(z)(1 - G_2(z))] \\
&= f'_{T }(z)(1 - G_1(z)) 
+ f'_{C }(z)(1 - G_2(z)) - f_{T }(z)G'_1(z) - f_{C }(z)G'_2(z), 
\end{aligned}
$$
it is  straightforward to obtain that
$$
\lim_{z \to \pm \infty}\frac{d \log \varpi (z)}{dz} 
=\lim_{z_1, z_2 \to \pm \infty}[\frac{d \log \omega (z_1, z_2)}{dz_1} 
+\frac{d \log \omega (z_1, z_2)}{dz_2}].
$$
Thus, Conditions (C\ref{C6}) and (C\ref{C7}) indicate that there exists a $T_{\tau} > N$ such that
$$
\psi_{\infty} 
\ge 
\sup_{z > T_{\tau}} \psi(H_0(z) - X^T \beta, H_0(z) - W^T \eta) 
\ge  
\inf_{z > T_{\tau}} \psi(H_0(z) - X^T \beta, H_0(z) - W^T \eta) \ge \psi_{\infty} - \tau.
$$
From Theorem \ref{theorem:consistency}, we have $|\hat{H}_n(t) - H_0(t)| = o_p(1)$ uniformly for $t \in [0, T_{\tau}]$. Then, using the abbreviation in \eqref{eq:note_abb}, for $z \in [0, T_n]$ we have
$$
\begin{aligned}
& |P_n I(Z \ge z) (\psi(\hat{\varepsilon}) - \psi(\varepsilon))| \\
& \leq P_n |I(Z > T_n) \psi(\varepsilon)| + |P_n I(Z \ge T_n) \psi(\hat{\varepsilon})| + |P_n I(Z \ge z, Z < T_{\tau})[\psi(\hat{\varepsilon}) - \psi(\varepsilon)]| \\
& \qquad + P_n I(Z \ge z, T_{\tau} < Z < T_n) | \psi(\hat{\varepsilon}) - \psi(\varepsilon)| \\
& \leq 2\psi_{\infty} P(Z > T_n) + c d_{T_{N}}(\hat{s}_n - s_0) + 2\tau P_n I(Z > z),
\end{aligned}
$$
where $c$ is a constant. Thus, equation \eqref{eq:lem:pres} becomes
$$
\begin{aligned}
\hat{H}_n(t) 
& = \int_{0}^{t} \frac{P_n dI(Z \leq z)}{P_n \{I(Z \geq z) \psi(\hat{\varepsilon})\}} \\
& = \int_{0}^{t} \frac{P_n dI(Z \leq z)}{P_n \{I(Z \geq z) \psi(\varepsilon)\}} - \int_{0}^{t} \frac{P_n dI(Z \geq z)(\psi(\hat{\varepsilon}) - \psi(\varepsilon))}{P_n \{I(Z \geq z) \psi(\varepsilon)\} P_n \{I(Z \geq z) \psi(\hat{\varepsilon})\}} P_n dI(Z \leq z).
\end{aligned}
$$
The function class $\{I(t \geq z): z \geq 0\}$ is the Glivenko–Cantelli class, which implies the first term on the right-hand side of the equation uniformly converges to $H_0(t)$ as $n \to \infty$. 
Consequently, for sufficiently large $n$, we have
$$
|\hat{H}_n(t) - H_0(t)| \leq 3\tau \int_{0}^{t} \frac{P I(Z \geq z)}{P \{I(Z \geq z) \psi(\varepsilon)\}} dH_0(z) + o_p(1).
$$
Thus, the proof is completed.

$\hfill\blacksquare$